\documentclass[journal]{IEEEtran}

\usepackage{cite}
\usepackage[OT1]{fontenc}
\usepackage[cmex10]{amsmath}
\usepackage{amssymb,amsfonts,amsthm}
\usepackage{bm}
\usepackage{braket}
\usepackage{graphicx}
\usepackage{color}
\usepackage{tabularx}
\usepackage{arydshln}
\usepackage{mathtools}
\usepackage{multirow}
\usepackage{algorithm,algpseudocode}
\usepackage{url}
\usepackage{listings}
\usepackage{comment}
\usepackage{CJKutf8}
\usepackage{tikz}
\usepackage{quantikz}
\usepackage{physics}
\usepackage{textcomp}
\usepackage[caption=false,font=footnotesize]{subfig} 

\def\BibTeX{{\rm B\kern-.05em{\sc i\kern-.025em b}\kern-.08em
    T\kern-.1667em\lower.7ex\hbox{E}\kern-.125emX}}

\newtheorem{lemma}{Lemma}
\newtheorem{theorem}{Theorem}
\newtheorem{definition}{Definition}

\newtheorem{proposition}{Proposition}

\title{Efficient Quantum Circuit Construction of Controlled Time-Evolution for Arbitrary Pauli-Sum Hamiltonians}

\author{Shintaro~Fujiwara,~\IEEEmembership{Graduate Student Member,~IEEE},
Naoki~Yamamoto, and
Naoki~Ishikawa,~\IEEEmembership{Senior Member,~IEEE}%
\thanks{S. Fujiwara and N. Ishikawa are with the Faculty of Engineering, Yokohama National University, 240-8501 Kanagawa, Japan (e-mail: fujiwara-shintaro-by@ynu.jp; ishikawa-naoki-fr@ynu.ac.jp).}%
\thanks{N. Yamamoto is with the Department of Applied Physics and Physico-Informatics, Keio University, 223-0061 Kanagawa, Japan (e-mail: yamamoto@appi.keio.ac.jp).}%
\thanks{This work was supported by JSPS KAKENHI Grant Numbers JP26KJ1221 and JP26K00944.
This work is also supported by the FY2025 MITOU Target Program, organized by the Information-technology Promotion Agency (IPA).}%
\thanks{Corresponding author: Shintaro Fujiwara (email: fujiwara-shintaro-by@ynu.jp).}}

\begin{document}

\maketitle

\begin{abstract}
Controlled time-evolution circuits select forward or backward Hamiltonian time evolution according to the state of an ancilla qubit. They are fundamental building blocks in quantum eigenvalue transformation of unitaries, Hamiltonian filtering, and related quantum algorithms. A direct realization adds ancilla control to the elementary gates of the time-evolution circuit and therefore increases the two-qubit gate count, compiled \(T\) depth and CX depth. We develop an efficient recursive algorithm that, for an arbitrary Pauli-sum Hamiltonian, partitions the input Pauli terms into groups and assigns to each group a sign-flip Pauli string that anti-commutes with the in-group terms, thereby removing ancilla control from the grouped time-evolution blocks. Numerical benchmarks on random Hamiltonians and structured spin Hamiltonians show reductions in compiled \(T\) depth and compiled CX depth. 
For a Kagome Hamiltonian with 24 spins under full connectivity, the proposed construction reduces the compiled \(T\) depth by \(85.2\%\) and the compiled CX depth by \(68.9\%\), compared with a conventional implementation that decomposes the Hamiltonian into individual Pauli terms and implements the controlled time evolution of each term by directly adding the ancilla qubit to the corresponding Pauli-rotation circuit.
\end{abstract}


\section{Introduction}
\label{sec:intro}

\IEEEPARstart{S}{everal} quantum algorithms use a one-ancilla circuit that selects between forward and backward real-time Hamiltonian evolution according to the ancilla state. Such a circuit appears in quantum eigenvalue transformation of unitaries (QETU), Hamiltonian filtering, and related state-preparation methods \cite{dong2022GroundState,karacan2025enhancing,dong2024multilevel,muller2024groundstate,karacan2025filter,kane2024nearly}. Since these algorithms call the controlled-evolution primitive repeatedly, reducing the depth of this primitive reduces the cost of the full algorithm.

In this paper we consider the signed controlled time-evolution circuit
\begin{equation}
D_H(t)
:=
\ketbra{0}{0}_{\mathcal A}\otimes e^{\mathrm{i} t H}
+
\ketbra{1}{1}_{\mathcal A}\otimes e^{-\mathrm{i} t H}.
\label{eq:intro_signed_controlled_evolution}
\end{equation}
Here, \(\mathcal A\) is the ancilla register and \(H\) acts on the system register. A direct realization controls the two time-evolution branches by the ancilla, as shown in Fig.~\ref{fig:intro}(a).

\begin{figure}[t]
  \centering
  \subfloat[Direct controlled realization.]{%
    \normalsize
    \begin{quantikz}[row sep={0.9cm,between origins}, column sep=0.9cm]
      \lstick{$\mathcal A$} & \octrl{1} & \ctrl{1} & \qw \\
      \lstick{$\mathcal S$} & \gate{e^{\mathrm{i}tH}} & \gate{e^{-\mathrm{i}tH}} & \qw
    \end{quantikz}
  }\\[2pt]
  \subfloat[Control-free realization based on Dong \textit{et al.}~\cite{dong2022GroundState}.]{%
    \normalsize
    \begin{quantikz}[row sep={0.9cm,between origins}, column sep=0.9cm]
      \lstick{$\mathcal A$} & \octrl{1} & \qw & \octrl{1} & \qw \\
      \lstick{$\mathcal S$} & \gate{K} & \gate{e^{-\mathrm{i}tH}} & \gate{K} & \qw
    \end{quantikz}
  }
  \caption{Two realizations of the signed controlled time evolution in Eq.~\eqref{eq:intro_signed_controlled_evolution}. The registers $\mathcal A$ and $\mathcal S$ are the ancilla and system registers, respectively. Open and filled control circles denote controls conditioned on $\ket{0}_{\mathcal A}$ and $\ket{1}_{\mathcal A}$, respectively. In (a), the ancilla directly controls the two time-evolution branches. In (b), the two occurrences of the Pauli string $K$ implement the conjugation $K e^{-\mathrm{i}tH}K=e^{\mathrm{i}tH}$. Consequently, the ancilla controls only $K$, while the time-evolution block itself remains uncontrolled.}
  \label{fig:intro}
\end{figure}
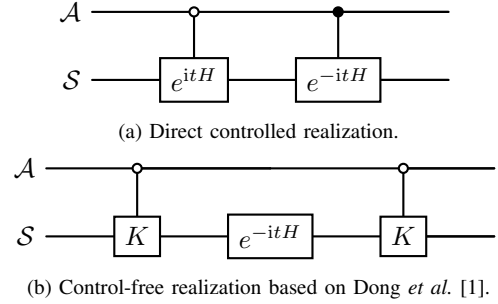

The need for such a primitive is not limited to QETU and Hamiltonian filtering. For example, probabilistic imaginary-time-evolution methods use one-ancilla forward and backward real-time-evolution subroutines~\cite{Kosugi2022Imaginary,nishi2023Optimal,xie2024probabilistic}. Out-of-time-ordered correlator (OTOC) measurement protocols also use time-reversal evolution subroutines around operator insertions~\cite{zhu2016measurement,gonzalezalonso2019outoftime,dressel2018Strengthening}. These examples motivate studying the circuit-level cost of controlled or time-reversal evolution primitives.

This direct realization of the circuit can be expensive after circuit synthesis and transpilation. Product-formula implementations of Hamiltonian time evolution consist of many elementary Pauli rotations \cite{Trotter1959On,suzuki1990Fractal}. Adding the same ancilla control to these rotations increases the number of two-qubit gates. It also restricts parallelism of quantum circuits. Specifically, rotations acting on disjoint system qubits may still share the ancilla after control is added. Consequently, the compiled \(T\) depth and CX depth can increase.

Dong \textit{et al.} introduced a way to avoid controlling the time-evolution block itself \cite{dong2022GroundState}. When a Pauli string $K$ reverses the sign of a Hamiltonian $H$ under conjugation, i.e., $KHK = -H$, an uncontrolled time-evolution block can be sandwiched by that Pauli string to select the sign of the evolution because $Ke^{-itH}K = e^{itH}$ holds. We refer to this Pauli string as a \emph{sign-flip Pauli string}. The ancilla then controls only the Pauli string, not the time-evolution circuit, as shown in Fig.~\ref{fig:intro}(b). 

Applying this idea to a general Pauli-sum Hamiltonian requires an additional algorithmic step. For a general Pauli-sum Hamiltonian, a single Pauli string may not be able to reverse the sign of all Pauli terms simultaneously. Therefore, the Hamiltonian must first be divided into term-groups, and one sign-flip Pauli string must be assigned to each group. Existing control-free constructions choose these groups and sign-flip Pauli strings from the special structure of the target model, such as spin or lattice Hamiltonians \cite{dong2022GroundState,karacan2025enhancing,muller2024groundstate,kane2024nearly}. In these constructions, the sign-flip strings are specified from model-dependent interaction patterns or symmetries. However, these works do not provide a general procedure that takes an arbitrary Pauli list as input and outputs both the term partition and the sign-flip string for each part.

We address this algorithmic problem by developing a recursive algorithm for arbitrary Pauli-sum Hamiltonians. The input is a list of non-identity Pauli terms. The output is a partition of that list into term-groups and one sign-flip Pauli string for each group. Each grouped time-evolution block can then be implemented without ancilla control. The ancilla controls only the corresponding sign-flip Pauli strings. The algorithm uses the binary symplectic representation of Pauli strings internally. 

The same signed primitive also yields a control-free construction of single-branch controlled Hamiltonian evolution. This connects the proposed construction to Hamiltonian simulation by quantum signal processing (QSP), phase-estimation-type procedures, generalized quantum signal processing (GQSP), and the Hadamard test \cite{low2017Optimal,kitaev1995Quantum,motlagh2024Generalized,berry2024Doubling,cleve1998Quantum}.

Pauli-string and phase-rotation circuit optimization has also been studied in compiler-level frameworks \cite{cowtan2020Phase,li2022Paulihedral,jin2024Tetris}. These methods optimize circuits built from Pauli-string rotations or phase-gadget structures themselves. In contrast, our work targets a controlled time-evolution primitive that calls a product-formula sequence of elementary Pauli evolutions conditioned on an ancilla qubit. The objective is therefore to remove the ancilla control from the Hamiltonian-evolution blocks, rather than to optimize a given uncontrolled Pauli-rotation circuit.
Other methods transform a Pauli-sum operator into a different representation involving pairwise anti-commuting Pauli strings, followed by a correction step \cite{schillo2026Block}. In contrast, our construction keeps the original Pauli terms and computes groups for which each group admits a sign-flip Pauli string. This is the structure needed to build controlled time-evolution circuits with uncontrolled evolution blocks.
Finally, related Pauli-grouping methods have different objectives. Commuting-group methods are used for measurement and circuit optimization \cite{jena2019Pauli,verteletskyi2020Measurement,yen2020Measuring}.

The contributions of this paper are summarized as follows.
\begin{itemize}
    \item We develop a recursive algorithm that removes ancilla control from time-evolution blocks for arbitrary Pauli-sum Hamiltonians. Given an input Pauli list, the algorithm partitions the terms into groups and computes one sign-flip Pauli string for each group. Using these data, each grouped time-evolution block is implemented without ancilla control, and the ancilla controls only the sign-flip Pauli strings, which can be implemented at low cost.
    \item We prove that the generated sign-flip Pauli string anti-commutes with every Pauli term in its group, and we derive structural properties and size bounds for the resulting term-groups.
    \item We evaluate the resulting circuits after Clifford+\(T\) synthesis and transpilation, and show reductions in compiled \(T\) depth and CX depth relative to direct and conventional term-by-term constructions.
    \item We show that the same signed primitive also yields a control-free construction of single-branch controlled Hamiltonian evolution.
\end{itemize}

The remainder of this paper is organized as follows. Section~\ref{sec:problem} introduces the notation and the binary symplectic representation of Pauli strings. Section~\ref{sec:construction} presents the recursive construction. Section~\ref{sec:grouping} proves its correctness and analyzes the resulting grouping from the viewpoint of circuit construction. Section~\ref{sec:numerical} reports numerical experiments. Section~\ref{sec:single_branch} discusses an application to single-branch controlled Hamiltonian evolution. Section~\ref{sec:conclusion} concludes the paper. 

\section{Problem formulation and binary symplectic notation}
\label{sec:problem}

As discussed in Section~\ref{sec:intro}, our task is to implement the signed controlled time-evolution subroutine in Eq.~\eqref{eq:intro_signed_controlled_evolution} for a general Pauli-sum Hamiltonian
\begin{equation}
H=\sum_{\ell=0}^{L-1} h_{\ell} P_{\ell},
\qquad
h_{\ell}\in\mathbb{R},
\qquad
P_{\ell}\in\{I,X,Y,Z\}^{\otimes n}.
\label{eq:problem_general_pauli_sum}
\end{equation}
If the Pauli-sum contains an identity component, we first combine all identity terms and write \(H=h_{\mathrm{id}}I^{\otimes n}+H_{\mathrm{ni}}\), where \(H_{\mathrm{ni}}\) contains only non-identity Pauli strings. Then
\begin{equation}
D_H(t)
=
\left(e^{\mathrm{i}t h_{\mathrm{id}}Z_{\mathcal A}}\otimes I\right)
D_{H_{\mathrm{ni}}}(t).
\label{eq:identity_component_ancilla_rotation}
\end{equation}
The factor \(e^{\mathrm{i}t h_{\mathrm{id}}Z_{\mathcal A}}\) is an ancilla-only relative phase between the two branches, not a global phase of \(D_H(t)\). Since it commutes with \(D_{H_{\mathrm{ni}}}(t)\), it may be inserted once before or after the circuit for \(D_{H_{\mathrm{ni}}}(t)\). Hence, in the sign-flip formulation below, \(H\) denotes the non-identity part.

The control-free construction of Dong \textit{et al.}~\cite{dong2022GroundState} is based on the conjugation identity
\begin{equation}
KHK=-H
\quad\Rightarrow\quad
K e^{-\mathrm{i}t H}K=e^{\mathrm{i}t H}.
\label{eq:dong_conjugation_identity}
\end{equation}
If a Pauli string \(K\) flips the sign of every non-identity Pauli term in \(H\), this identity removes the ancilla control from the time-evolution block. For a general Pauli-sum Hamiltonian, however, there may be no single Pauli string with this property. We therefore partition the Pauli terms into term-groups and assign to each term-group a Pauli string that flips the sign of the corresponding partial Hamiltonian. In this section, we first formulate this grouped sign-flip synthesis problem and then restate it in the binary symplectic representation used throughout Sections~\ref{sec:construction} and~\ref{sec:grouping}.

\subsection{Overview of the Grouped Control-free Construction}
\label{subsec:grouped_formulation}

We introduce a positive integer $G$ and a map
\begin{equation}
\gamma:\{0,\dots,L-1\}\to\{0,\dots,G-1\},
\label{eq:gamma}
\end{equation}
which assigns a term-group index to each Pauli term in Eq.~\eqref{eq:problem_general_pauli_sum}. For each $g\in\{0,\dots,G-1\}$, define
\begin{equation}
\mathcal I_g
:=
\{\ell\in\{0,\dots,L-1\}\mid \gamma(\ell)=g\}.
\label{eq:problem_group_index_set}
\end{equation}
Then
\begin{equation}
\bigcup_{g=0}^{G-1}\mathcal I_g=\{0,\dots,L-1\},
\label{eq:problem_group_cover}
\end{equation}
and
\begin{equation}
\mathcal I_g\cap\mathcal I_{g'}=\varnothing
\qquad
(g\neq g').
\label{eq:problem_group_disjoint}
\end{equation}
Using these index sets, we define the partial Hamiltonians
\begin{equation}
H_g
:=
\sum_{\ell\in\mathcal I_g} h_\ell P_\ell
\qquad
0\le g\le G-1,
\label{eq:problem_grouped_hamiltonian}
\end{equation}
which gives
\begin{equation}
H=\sum_{g=0}^{G-1} H_g.
\label{eq:problem_grouped_sum}
\end{equation}

For each term-group $g$, we seek a Pauli string $K_g$ such that
\begin{equation}
\{K_g,P_\ell\}=0
\qquad
\text{for every }\ell\in\mathcal I_g.
\label{eq:problem_groupwise_anticommutation}
\end{equation}
We call such a Pauli string $K_g$ a sign-flip Pauli string for the term-group $g$. Because $h_\ell\in\mathbb R$ and each $P_\ell$ is Hermitian, $H_g$ is Hermitian. Moreover, Eq.~\eqref{eq:problem_groupwise_anticommutation} gives
\begin{equation}
\begin{split}
K_g H_g K_g
&=
\sum_{\ell\in\mathcal I_g} h_\ell K_g P_\ell K_g\\
&=
-\sum_{\ell\in\mathcal I_g} h_\ell P_\ell
=
-H_g.
\end{split}
\label{eq:problem_groupwise_sign_flip}
\end{equation}
Consequently,
\begin{equation}
K_g e^{-\mathrm{i}t H_g} K_g=e^{\mathrm{i}t H_g}
\label{eq:problem_groupwise_evolution_flip}
\end{equation}
for every real $t$.

Define the ancilla-controlled Pauli operation, where the superscript $(0)$ indicates control on the ancilla state $\ket{0}_{\mathcal A}$,
\begin{equation}
C^{(0)}(K_g)
:=
\ketbra{0}{0}_{\mathcal A}\otimes K_g
+
\ketbra{1}{1}_{\mathcal A}\otimes I.
\label{eq:problem_control_zero_pauli}
\end{equation}
Then Eq.~\eqref{eq:problem_groupwise_evolution_flip} gives
\begin{equation}
\begin{split}
&C^{(0)}(K_g)
\left(I_{\mathcal A}\otimes e^{-\mathrm{i}t H_g}\right)
C^{(0)}(K_g)
\\
=&\ketbra{0}{0}_{\mathcal A}\otimes e^{\mathrm{i}t H_g}
+\ketbra{1}{1}_{\mathcal A}\otimes e^{-\mathrm{i}t H_g}.
\end{split}
\label{eq:problem_groupwise_controlfree_realization}
\end{equation}
Hence, the ancilla control is applied only to the Pauli string
\(K_g\), while the time-evolution block itself remains uncontrolled.
The identities above are exact operator identities for the target
evolution blocks. In circuit implementations, these blocks are replaced
by product-formula approximations built from elementary Pauli
evolutions. Accordingly, the grouped sign-flip synthesis problem studied
in this paper is the following. Given the Pauli terms in
Eq.~\eqref{eq:problem_general_pauli_sum}, compute a map \(\gamma\)
and Pauli strings \(K_0,\dots,K_{G-1}\) such that
Eq.~\eqref{eq:problem_groupwise_evolution_flip} holds for every \(g\).
The product-formula realization of the resulting grouped evolution is
discussed in Section~\ref{subsec:grouped_control_free_product_formula}. 

Fig.~\ref{fig:proposed_method_flow} summarizes the overall flow of the proposed grouped control-free construction. Starting from a Pauli-sum Hamiltonian, the recursive algorithm is applied to each input Pauli term to obtain a group label. Pauli terms with the same group label are collected into the same term-group, and the group label determines the sign-flip Pauli string assigned to that group. The grouped data are finally converted into a signed controlled time-evolution circuit in which ancilla control appears only on the sign-flip Pauli strings. The details of the algorithm are explained in Section~\ref{sec:construction}.

\begin{figure}[t]
  \centering
  \subfloat[Target Hamiltonian.]{%
    \includegraphics[scale=0.75]{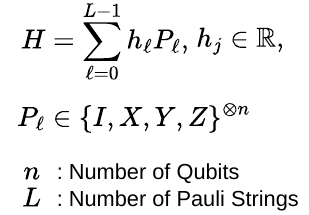}%
  }\\
  \subfloat[Group-label computation.]{%
    \includegraphics[scale=0.75]{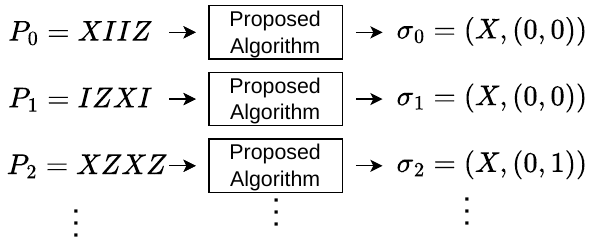}%
  }\\
  \subfloat[Term grouping.]{%
    \includegraphics[scale=0.75]{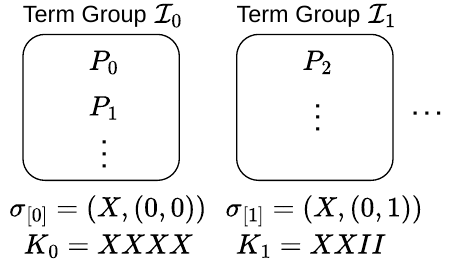}%
  }\\
  \subfloat[Signed controlled time evolution.]{%
    \includegraphics[scale=0.75]{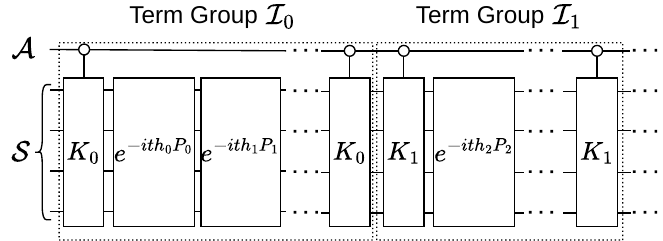}%
  }
  \caption{Overall flow of the proposed grouped control-free construction. The input Hamiltonian is first processed term by term, and the resulting data are then assembled into the final signed controlled time-evolution circuit. (a) The starting point is a Pauli-sum Hamiltonian written as a linear combination of Pauli strings. (b) The proposed recursive algorithm is applied to each Pauli term to compute its group label. (c) Pauli terms with the same group label are collected into the same term-group, and that label determines the sign-flip Pauli string assigned to the group. (d) The grouped data are finally converted into a signed controlled time-evolution circuit in which the ancilla controls only the sign-flip Pauli strings associated with the term-groups.}
  \label{fig:proposed_method_flow}
\end{figure}

In the next subsection, we restate this grouped sign-flip synthesis problem in binary symplectic form, which is the form used by the recursive construction.

\subsection{Binary Symplectic Notation}
\label{subsec:binary_symplectic_notation}

We use the standard binary symplectic representation of Pauli strings~\cite{gottesman1997Stabilizer,dehaene2003Clifford}. Let
\begin{equation}
V:=\mathbb F_2^2,
\label{eq:problem_V_definition}
\end{equation}
which is the two-dimensional vector space over the binary field $\mathbb F_2$.
For
\begin{equation}
u=(u_X,u_Z) \in V,
\qquad
v=(v_X,v_Z)\in V,
\label{eq:problem_V_elements}
\end{equation}
we define addition on $V$ entry-wise by
\begin{equation}
u+v:=(u_X+v_X,\,u_Z+v_Z)\pmod 2.
\label{eq:problem_V_addition}
\end{equation}
We denote the zero vector of $V$ by
\begin{equation}
\bm{0}:=(0,0).
\label{eq:problem_V_zero}
\end{equation}

For Pauli operators acting on a single qubit, define 
\begin{equation}
P(u) := \mathrm{i}^{u_X u_Z} X^{u_X} Z^{u_Z},
\qquad
u=(u_X,u_Z)\in V.
\label{eq:problem_one_qubit_pauli_encoding}
\end{equation}
Then
\begin{equation}
\begin{split}
P(0,0)=I,
\qquad
P(1,0)=X,\\
P(0,1)=Z,
\qquad
P(1,1)=Y.
\end{split}
\label{eq:problem_one_qubit_pauli_table}
\end{equation}

For an $n$-qubit Pauli string, write
\begin{equation}
\bm{p}=(u_0,\dots,u_{n-1})\in V^n,
\label{eq:problem_V_word}
\end{equation}
and define
\begin{equation}
P(\bm{p}):=\bigotimes_{j=0}^{n-1} P(u_j).
\label{eq:problem_word_to_pauli}
\end{equation}
Equivalently, if
\begin{equation}
\bm{x}=(x_0,\dots,x_{n-1}),
\qquad
\bm{z}=(z_0,\dots,z_{n-1})\in\mathbb F_2^n,
\label{eq:problem_xz_vectors}
\end{equation}
then
\begin{equation}
\bm{p}=(\bm{x}\mid\bm{z})=(x_0,\dots,x_{n-1}\mid z_0,\dots,z_{n-1})\in\mathbb F_2^{2n}
\label{eq:problem_binary_symplectic_vector}
\end{equation}
means that $u_j=(x_j,z_j)$ for every $j\in\{0,\dots,n-1\}$. Thus we identify the $V$-valued word space $V^n$ with the binary vector space $\mathbb F_2^{2n}$ through the correspondence
\begin{equation}
(u_0,\dots,u_{n-1})
\longleftrightarrow
(\bm{x}\mid\bm{z}),
\qquad
u_j=(x_j,z_j).
\label{eq:problem_Vn_F2_identification}
\end{equation}

For
\begin{equation}
\bm{p}=(u_0,\dots,u_{n-1}),
\qquad
\bm{q}=(v_0,\dots,v_{n-1})\in V^n,
\label{eq:problem_Vn_elements}
\end{equation}
we define addition on $V^n$ entry-wise by
\begin{equation}
\bm{p}+\bm{q}:=(u_0+v_0,\dots,u_{n-1}+v_{n-1})\in V^n,
\label{eq:problem_Vn_addition}
\end{equation}
where each $u_j+v_j$ is the addition in $V$ defined in Eq.~\eqref{eq:problem_V_addition}. We denote the zero vector of $V^n$ by
\begin{equation}
\bm{0}_n:=(\bm{0},\dots,\bm{0})\in V^n.
\label{eq:problem_Vn_zero}
\end{equation}
The support of $P(\bm{p})$ is
\begin{equation}
\operatorname{supp}(P(\bm{p}))
:=
\{j\in\{0,\dots,n-1\}\mid u_j\neq \bm{0}\}.
\label{eq:problem_support_definition}
\end{equation}

For two one-qubit vectors $u=(u_X,u_Z)$ and $v=(v_X,v_Z)$ in $V$, define the local symplectic form
\begin{equation}
\omega:V\times V\to\mathbb F_2,
\qquad
\omega(u,v) := u_X v_Z + u_Z v_X \pmod 2.
\label{eq:problem_local_symplectic_form}
\end{equation}
For two $n$-qubit words
\begin{equation}
\bm{p}=(u_0,\dots,u_{n-1}),
\qquad
\bm{q}=(v_0,\dots,v_{n-1})\in V^n,
\label{eq:problem_global_words}
\end{equation}
define the global symplectic inner product by
\begin{equation}
\begin{split}
\Omega:V^n\times V^n&\to\mathbb F_2,\\
\qquad
\Omega(\bm{p},\bm{q})
:=
\sum_{j=0}^{n-1} \omega(u_j,v_j)
&=
\bm{x}\cdot\bm{z}' + \bm{z}\cdot\bm{x}'
\pmod 2,
\end{split}
\label{eq:problem_global_symplectic_form}
\end{equation}
where $\bm{p}=(\bm{x}\mid\bm{z})$ and $\bm{q}=(\bm{x}'\mid\bm{z}')$, and
\begin{equation}
\bm{a}\cdot\bm{b}:=\sum_{j=0}^{n-1} a_j b_j \pmod 2
\qquad
(\bm{a},\bm{b}\in\mathbb F_2^n).
\label{eq:problem_F2_dot_product}
\end{equation}
The map $\Omega$ is bilinear over $\mathbb F_2$, and $\Omega(\bm{0}_n,\bm{p})=\Omega(\bm{p},\bm{0}_n)=0$ for every $\bm{p}\in V^n$.

With this notation, the commutation relation between two Pauli strings is written as~\cite{dehaene2003Clifford}
\begin{equation}
P(\bm{p})P(\bm{q})=(-1)^{\Omega(\bm{p},\bm{q})}P(\bm{q})P(\bm{p}).
\label{eq:problem_commutation_symplectic}
\end{equation}
Therefore, $P(\bm{p})$ and $P(\bm{q})$ commute when
\begin{equation}
\Omega(\bm{p},\bm{q})=0
\iff
[P(\bm{p}),P(\bm{q})]=0,
\label{eq:problem_commutation_condition}
\end{equation}
and anti-commute when
\begin{equation}
\Omega(\bm{p},\bm{q})=1
\iff
\{P(\bm{p}),P(\bm{q})\}=0.
\label{eq:problem_anticommutation_condition}
\end{equation}
Accordingly, the grouped anti-commutation problem in Eq.~\eqref{eq:problem_groupwise_anticommutation} becomes a binary-symplectic compatibility problem: for each $g$, find a binary vector $\bm{k}_g\in\mathbb F_2^{2n}$ such that
\begin{equation}
\Omega(\bm{k}_g,\bm{p}_\ell)=1
\qquad
\text{for all }\ell\in\mathcal I_g.
\label{eq:problem_binary_groupwise_condition}
\end{equation}
The recursive construction in Section~\ref{sec:construction} is designed to compute such vectors group-wise from the input Pauli list.

\section{Recursive construction of term-groups and sign-flip Pauli strings}
\label{sec:construction}

We now describe the recursive algorithm. The algorithm mainly consists of two parts, \emph{Reduction} and \emph{Generation}. The input is the list of non-identity Pauli terms \(P_0,\dots,P_{L-1}\) on \(n\) qubits. The output consists of a partition $\{\mathcal I_g\}_{g=0}^{G-1}$ of the term-index set \(\{0,\dots,L-1\}\), together with one Pauli string \(K_g\) for each set \(\mathcal I_g\). 

We describe the algorithm below.

\subsection{Reduction on Padded \(V\)-valued Words}
\label{subsec:construction_reduction}

For each non-identity Pauli term \(P_\ell\), let
\begin{equation}
\bm{p}_\ell=(u_{\ell,0},\dots,u_{\ell,n-1})\in V^n
\label{eq:construction_input_word}
\end{equation}
be its \(V\)-valued word. Let
\begin{equation}
m:=\lceil \log_2 n\rceil,
\qquad
N:=2^m.
\label{eq:construction_power_of_two_length}
\end{equation}
We pad \(\bm{p}_\ell\) with zero symbols to length \(N\) and define
\begin{equation}
\bm{p}_\ell^{(m)}
:=
(u_{\ell,0},\dots,u_{\ell,n-1},\bm{0},\dots,\bm{0})\in V^N.
\label{eq:construction_zero_padded_word}
\end{equation}

\emph{Reduction} maps a nonzero padded word to two pieces of classical data: a control-bit sequence and a root symbol. 

\begin{definition}[Reduction]
\label{def:reduction_general_pauli}
Let \(\bm{p}^{(m)}\in V^{2^m}\) be nonzero. For each \(j=m-1,m-2,\dots,0\), write
\begin{equation}
\bm{p}^{(j+1)}=[\bm{a}^{(j)};\bm{b}^{(j)}],
\qquad
\bm{a}^{(j)},\bm{b}^{(j)}\in V^{2^j},
\label{eq:construction_reduction_split}
\end{equation}
where \([\cdot;\cdot]\) denotes concatenation of the upper and lower halves. For \(\bm{a}=(a_0,\dots,a_{2^j-1})\) and \(\bm{b}=(b_0,\dots,b_{2^j-1})\), this means
\begin{equation}
[\bm{a};\bm{b}]
=
(a_0,\dots,a_{2^j-1},b_0,\dots,b_{2^j-1})
\in V^{2^{j+1}}.
\label{eq:construction_concatenation_definition}
\end{equation}
Define
\begin{equation}
c_j
:=
\begin{cases}
1, & \bm{a}^{(j)}=\bm{b}^{(j)},\\
0, & \bm{a}^{(j)}\neq\bm{b}^{(j)},
\end{cases}
\qquad
0\le j\le m-1.
\label{eq:construction_control_bit_rule}
\end{equation}
We call \(c_j\) the control bit at level \(j\). The next reduced word is defined by
\begin{equation}
\bm{p}^{(j)}
:=
\begin{cases}
\bm{a}^{(j)}, & c_j=1,\\
\bm{a}^{(j)}+\bm{b}^{(j)}, & c_j=0.
\end{cases}
\label{eq:construction_reduction_rule}
\end{equation}
Equal halves are represented by one copy. Unequal halves are replaced by their entry-wise sum over \(V\). This rule preserves the symplectic inner product with the generated word at each reduction level, as proved in Lemma~\ref{lem:grouping_stagewise_invariance}.
The resulting bit sequence is denoted by
\begin{equation}
\bm{c}(\bm{p}^{(m)}) := (c_0,\dots,c_{m-1})\in\mathbb F_2^m.
\label{eq:construction_control_bits_definition}
\end{equation}
The final one-symbol word is called the root symbol and is denoted by
\begin{equation}
r(\bm{p}^{(m)}) := \bm{p}^{(0)}\in V.
\label{eq:construction_root_symbol_definition}
\end{equation}
\end{definition}
The superscript in \(\bm{p}^{(j)}\) denotes the reduction level. The word \(\bm{p}^{(j)}\) has length \(2^j\).

For the non-identity Pauli terms considered here, the next proposition shows that the root symbol is always nonzero.

\begin{proposition}[Nonzero root for non-identity inputs]
\label{prop:construction_nonzero_root}
Let \(\bm{p}^{(m)}\in V^{2^m}\) be nonzero, and let
\begin{equation}
\bm{p}^{(0)},\bm{p}^{(1)},\dots,\bm{p}^{(m)}
\label{eq:construction_nonzero_sequence}
\end{equation}
be the Reduction sequence from Definition~\ref{def:reduction_general_pauli}. Then
\begin{equation}
\bm{p}^{(j)}\neq \bm{0}_{2^j}
\qquad
\text{for every }j\in\{0,\dots,m\}.
\label{eq:construction_nonzero_conclusion}
\end{equation}
In particular,
\begin{equation}
r(\bm{p}^{(m)})=\bm{p}^{(0)}\in V\setminus\{\bm{0}\}.
\label{eq:construction_nonzero_root}
\end{equation}
\end{proposition}

\begin{proof}
Suppose, for contradiction, that \(\bm{p}^{(t)}=\bm{0}_{2^t}\) for some \(t\), and choose the largest such \(t\). Then \(t<m\), because \(\bm{p}^{(m)}\neq\bm{0}_{2^m}\). Write
\begin{equation}
\bm{p}^{(t+1)}=[\bm{a}^{(t)};\bm{b}^{(t)}].
\label{eq:construction_nonzero_split}
\end{equation}
Since \(\bm{p}^{(t+1)}\neq\bm{0}_{2^{t+1}}\), at least one of \(\bm{a}^{(t)}\) and \(\bm{b}^{(t)}\) is nonzero.

If \(c_t=1\), then \(\bm{a}^{(t)}=\bm{b}^{(t)}\) and \(\bm{p}^{(t)}=\bm{a}^{(t)}\). Hence \(\bm{p}^{(t)}=\bm{0}_{2^t}\) implies \(\bm{a}^{(t)}=\bm{b}^{(t)}=\bm{0}_{2^t}\), which contradicts \(\bm{p}^{(t+1)}\neq\bm{0}_{2^{t+1}}\).

If \(c_t=0\), then \(\bm{p}^{(t)}=\bm{a}^{(t)}+\bm{b}^{(t)}\). Therefore \(\bm{p}^{(t)}=\bm{0}_{2^t}\) implies \(\bm{a}^{(t)}=\bm{b}^{(t)}\), which contradicts the defining condition \(c_t=0\).

Both cases contradict the assumptions on \(t\). Hence no such \(t\) exists.
\end{proof}

Since the root symbol is nonzero, it can be mapped to a one-qubit symbol that anti-commutes with $P(r(\bm{p}^{(m)}))$. This symbol and the control-bit sequence form the group label.

\begin{definition}[Group label]
\label{def:construction_group_label}
Fix a map
\begin{equation}
\kappa:V\setminus\{\bm{0}\}\to V\setminus\{\bm{0}\}
\label{eq:construction_kappa_map}
\end{equation}
satisfying
\begin{equation}
\omega\bigl(\kappa(r),r\bigr)=1
\qquad
\text{for every }r\in V\setminus\{\bm{0}\}.
\label{eq:construction_kappa_anticommutes}
\end{equation}
For a nonzero input word \(\bm{p}^{(m)}\) arising from a non-identity Pauli term, define its seed symbol by
\begin{equation}
s(\bm{p}^{(m)})
:=
\kappa(r(\bm{p}^{(m)}))
\in
\operatorname{Im}(\kappa).
\label{eq:construction_seed_symbol_definition}
\end{equation}
The group label is the pair
\begin{equation}
\sigma(\bm{p}^{(m)})
:=
\bigl(s(\bm{p}^{(m)}),\bm{c}(\bm{p}^{(m)})\bigr)
\in
\operatorname{Im}(\kappa)\times\mathbb F_2^m.
\label{eq:construction_group_label_definition}
\end{equation}
\end{definition}

For the concrete construction used below, we adopt the choice
\begin{equation}
\kappa(1,0)=(1,1),
\qquad
\kappa(0,1)=(1,1),
\qquad
\kappa(1,1)=(1,0),
\label{eq:construction_kappa_choice}
\end{equation}
which corresponds to
\begin{equation}
\kappa(X)=Y,
\qquad
\kappa(Z)=Y,
\qquad
\kappa(Y)=X.
\label{eq:construction_kappa_choice_pauli}
\end{equation}
The image of an admissible map \(\kappa\) cannot consist of a single symbol.
Indeed, suppose that \(|\operatorname{Im}(\kappa)|=1\). Then there exists
\(s\in\{X,Z,Y\}\) such that
\begin{equation}
\kappa(X)=\kappa(Z)=\kappa(Y)=s.
\end{equation}

The choice of \(\kappa\) in
Eq.~\eqref{eq:construction_kappa_choice} uses the minimum possible
number of seed symbols. For an admissible map, the allowed images are
\begin{equation}
\kappa(X)\in\{Z,Y\},\qquad
\kappa(Z)\in\{X,Y\},\qquad
\kappa(Y)\in\{X,Z\}.
\end{equation}
No single Pauli symbol belongs to all three sets. Hence an admissible
map cannot satisfy \(|\operatorname{Im}(\kappa)|=1\), and therefore
\begin{equation}
|\operatorname{Im}(\kappa)|\ge 2 .
\end{equation}
The choice in Eq.~\eqref{eq:construction_kappa_choice} has image
\(\{X,Y\}\), and therefore this lower bound is attained.

\subsection{Generation from a group label}
\label{subsec:construction_generation}

\emph{Generation} starts from the seed symbol in the group label generated in Reduction and expands it upward according to the recorded control-bit sequence.

\begin{definition}[Generation from a group label]
\label{def:generation_general_pauli}
Let
\begin{equation}
\sigma=(s,\bm{c}),
\qquad
s\in\operatorname{Im}(\kappa),
\quad
\bm{c}=(c_0,\dots,c_{m-1})\in\mathbb F_2^m.
\label{eq:construction_generation_input}
\end{equation}
The Generation algorithm starts from
\begin{equation}
\bm{k}^{(0)}:=s\in V
\label{eq:construction_generation_initial}
\end{equation}
and recursively defines \(\bm{k}^{(j+1)}\in V^{2^{j+1}}\) by
\begin{equation}
\bm{k}^{(j+1)}
:=
\begin{cases}
[\bm{k}^{(j)};\bm{0}_{2^j}], & c_j=1,\\
[\bm{k}^{(j)};\bm{k}^{(j)}], & c_j=0,
\end{cases}
\qquad
0\le j\le m-1.
\label{eq:construction_generation_rule}
\end{equation}
The final word \(\bm{k}^{(m)}\in V^{2^m}\) is called the generated sign-flip word associated with \(\sigma\).
\end{definition}

Since the original Hamiltonian acts on \(n\) qubits rather than on the padded length \(N\), we discard the padded tail of the generated word. After truncating \(\bm{k}^{(m)}\) to the first \(n\) entries, we obtain a word in \(V^n\). We denote this truncated word by \(\bm{k}(\sigma)\). 
Through the identification \(V^n\simeq\mathbb F_2^{2n}\) in Section~\ref{subsec:binary_symplectic_notation}, it is also viewed as a binary symplectic vector. 
We write
\begin{equation}
K(\sigma):=P(\bm{k}(\sigma)).
\label{eq:construction_K_sigma_definition}
\end{equation}

Fig.~\ref{fig:algorithm_flow} shows the data flow of Reduction and Generation for one Pauli term.

\begin{figure}[t]
  \centering
  \subfloat[Mapping to a binary-symplectic word.]{%
    \includegraphics[scale=0.75]{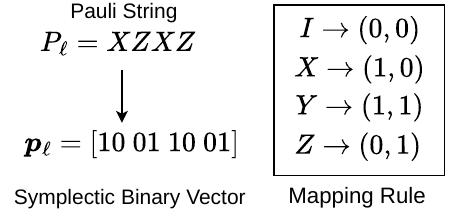}%
  }\\
  \subfloat[Reduction.]{%
    \includegraphics[scale=0.69]{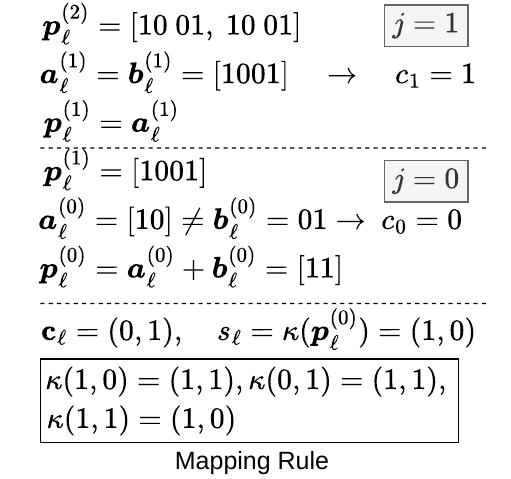}%
  }\\
  \subfloat[Generation.]{%
    \includegraphics[scale=0.69]{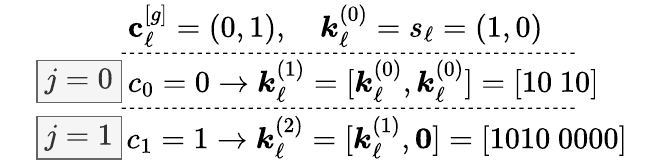}%
  }\\
  \subfloat[Mapping back to a Pauli string.]{%
    \includegraphics[scale=0.75]{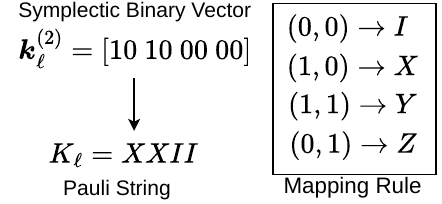}%
  }
    \caption{Internal flow of the recursive construction for one Pauli term. The Pauli string is converted to a binary-symplectic word, Reduction produces a control-bit sequence and a root symbol, Generation constructs the sign-flip word from the corresponding group label, and the result is mapped back to a Pauli string. The mapping rule used to convert the root symbol into the seed symbol is the concrete map \(\kappa\) specified in Eqs.~\eqref{eq:construction_kappa_choice} and~\eqref{eq:construction_kappa_choice_pauli}. The displayed example uses \(P_\ell=XZXZ\) with \(n=4\) and \(N=4\), and no zero padding is present.}
  \label{fig:algorithm_flow}
\end{figure}

\subsection{Term-groups and Their Sign-flip Pauli strings}
\label{subsec:construction_term_groups}

The group label assigned to each input term determines the term-group partition. Terms with the same label are placed in the same group. Since the algorithm groups term indices, each term-group is a subset of \(\{0,\dots,L-1\}\).

Applying Reduction to each padded input word \(\bm{p}_\ell^{(m)}\) gives a group label
\begin{equation}
\sigma_\ell
:=
\sigma(\bm{p}_\ell^{(m)})
=
(s_\ell,\bm{c}_\ell),
\label{eq:construction_each_label}
\end{equation}
where
\begin{equation}
s_\ell=\kappa\bigl(r(\bm{p}_\ell^{(m)})\bigr).
\label{eq:construction_each_seed_symbol}
\end{equation}
Let
\begin{equation}
\Sigma
:=
\{\sigma_\ell\mid 0\le \ell\le L-1\}
\label{eq:construction_distinct_labels}
\end{equation}
be the set of distinct group labels that appear among the input terms, and let
\begin{equation}
G:=|\Sigma|.
\label{eq:construction_number_of_groups}
\end{equation}
be the number of distinct labels.
Fix an enumeration
\begin{equation}
\Sigma=\{\sigma_{[0]},\dots,\sigma_{[G-1]}\},
\label{eq:construction_enumerated_seed_labels}
\end{equation}
where the elements \(\sigma_{[0]},\dots,\sigma_{[G-1]}\) are pairwise distinct. We then define an equivalence relation 
\begin{equation}
\gamma(\ell)=g
\qquad
\Longleftrightarrow
\qquad
\sigma_\ell=\sigma_{[g]},
\label{eq:construction_gamma_from_seed_label}
\end{equation}
where the map $\gamma$ is defined as Eq.~\eqref{eq:gamma}.
For each \(g\in\{0,\dots,G-1\}\), the associated term-group is the index set
\begin{equation}
\mathcal I_g
:=
\{\ell\in\{0,\dots,L-1\}\mid \gamma(\ell)=g\}.
\label{eq:construction_group_indices}
\end{equation}

For each distinct group label \(\sigma_{[g]}\), apply Generation and truncate the resulting word to the first \(n\) entries. This yields
\begin{equation}
\bm{k}_g:=\bm{k}(\sigma_{[g]})\in V^n
\label{eq:construction_generated_binary_vector}
\end{equation}
and the corresponding sign-flip Pauli string
\begin{equation}
K_g:=K(\sigma_{[g]})=P(\bm{k}_g).
\label{eq:construction_group_pauli}
\end{equation}
The overall algorithm is summarized in Algorithm~\ref{alg:construction_grouping}.

\begin{algorithm}[t]
\caption{Recursive term-group construction and sign-flip assignment}
\label{alg:construction_grouping}
\begin{algorithmic}[1]
\Require Non-identity Pauli terms \(P_0,\dots,P_{L-1}\) on \(n\) qubits.
\Ensure A partition \(\{\mathcal I_g\}_{g=0}^{G-1}\) of \(\{0,\dots,L-1\}\) and Pauli strings \(K_0,\dots,K_{G-1}\).
\State Set \(m\gets \lceil \log_2 n\rceil\) and \(N\gets 2^m\).
\For{\(\ell=0\) to \(L-1\)}
    \State Convert \(P_\ell\) to the \(V\)-valued word \(\bm{p}_\ell=(u_{\ell,0},\dots,u_{\ell,n-1})\).
    \State Pad \(\bm{p}_\ell\) with zeros to length \(N\) and denote the result by \(\bm{p}_\ell^{(m)}\).
    \For{\(j=m-1\) down to \(0\)}
        \State Split \(\bm{p}_\ell^{(j+1)}\) as \([\bm{a}_\ell^{(j)};\bm{b}_\ell^{(j)}]\).
        \If{\(\bm{a}_\ell^{(j)}=\bm{b}_\ell^{(j)}\)}
            \State \(\bm{p}_\ell^{(j)}\gets \bm{a}_\ell^{(j)}\).
            \State \(c_{\ell,j}\gets 1\).
        \Else
            \State \(\bm{p}_\ell^{(j)}\gets \bm{a}_\ell^{(j)}+\bm{b}_\ell^{(j)}\).
            \State \(c_{\ell,j}\gets 0\).
        \EndIf
    \EndFor
    \State Set \(\bm{c}_\ell\gets (c_{\ell,0},\dots,c_{\ell,m-1})\).
    \State Set \(s_\ell\gets \kappa(\bm{p}_\ell^{(0)})\).
    \State Set \(\sigma_\ell\gets (s_\ell,\bm{c}_\ell)\).
\EndFor
\State Form the set \(\Sigma=\{\sigma_\ell\mid 0\le \ell\le L-1\}\) and enumerate it as \(\{\sigma_{[0]},\dots,\sigma_{[G-1]}\}\).
\For{\(g=0\) to \(G-1\)}
    \State Define \(\mathcal I_g\gets \{\ell\mid \sigma_\ell=\sigma_{[g]}\}\).
    \State Write \(\sigma_{[g]}=(s_{[g]},\bm{c}_{[g]})\), where \(\bm{c}_{[g]}=(c_{{[g],0}},\dots,c_{[g],m-1})\).
    \State Initialize \(\bm{k}_g^{(0)}\gets s_{[g]}\).
    \For{\(j=0\) to \(m-1\)}
        \If{\(c_{{[g]},j}=1\)}
            \State \(\bm{k}_g^{(j+1)}\gets [\bm{k}_g^{(j)};\bm{0}_{2^j}]\).
        \Else
            \State \(\bm{k}_g^{(j+1)}\gets [\bm{k}_g^{(j)};\bm{k}_g^{(j)}]\).
        \EndIf
    \EndFor
    \State Truncate \(\bm{k}_g^{(m)}\) to its first \(n\) entries and denote the result by \(\bm{k}_g\in V^n\).
    \State Define \(K_g\gets P(\bm{k}_g)\).
\EndFor
\State Define \(\gamma(\ell)=g\) whenever \(\sigma_\ell=\sigma_{[g]}\).
\State \Return \(\{\mathcal I_g\}_{g=0}^{G-1}\) and \(K_0,\dots,K_{G-1}\).
\end{algorithmic}
\end{algorithm}


The output of Algorithm~\ref{alg:construction_grouping} gives the decomposition
\begin{equation}
H=\sum_{g=0}^{G-1} H_g,
\qquad
H_g:=\sum_{\ell\in\mathcal I_g} h_\ell P_\ell,
\label{eq:construction_grouped_hamiltonian}
\end{equation}
together with one sign-flip Pauli string \(K_g\) for each term-group. This completes the core recursive construction used in the remainder of the paper.

\subsection{Example}
\label{ex:construction_running_xyz}

We illustrate the construction on the \(4\)-qubit Hamiltonian
\begin{equation}
H
=
h_0\,XIIZ
+
h_1\,IZXI
+
h_2\,XZXZ
+
h_3\,YXZY.
\label{eq:construction_running_xyz_hamiltonian}
\end{equation}
Here \(n=4\), hence
\begin{equation}
m=\lceil \log_2 4\rceil=2,
\qquad
N=2^m=4,
\end{equation}
and therefore no zero padding is needed. In this example, we write the nonzero elements of \(V\) by the corresponding one-qubit Pauli symbols \(X,Z,Y\), and we write the zero element of \(V\) as \(0\).

Using the concrete map in Eq.~\eqref{eq:construction_kappa_choice}, in particular \(\kappa(Y)=X\),
Reduction gives the labels shown in Table~\ref{tab:construction_example_labels}.

\begin{table}[t]
\centering
\caption{Reduction output for the running example.}
\label{tab:construction_example_labels}
\small
\begin{tabular}{c c c c c}
\hline
\(\ell\) & \(P_\ell\) & \(r(\bm{p}_\ell^{(2)})\) & \(\bm{c}_\ell\) & \(\sigma_\ell\) \\
\hline
\(0\) & \(XIIZ\) & \(Y\) & \((0,0)\) & \((X,(0,0))\) \\
\(1\) & \(IZXI\) & \(Y\) & \((0,0)\) & \((X,(0,0))\) \\
\(2\) & \(XZXZ\) & \(Y\) & \((0,1)\) & \((X,(0,1))\) \\
\(3\) & \(YXZY\) & \(Y\) & \((0,0)\) & \((X,(0,0))\) \\
\hline
\end{tabular}
\end{table}

Hence there are two distinct labels,
\begin{equation}
\sigma_{[0]}=(X,(0,0)),
\qquad
\sigma_{[1]}=(X,(0,1)).
\end{equation}
Algorithm~\ref{alg:construction_grouping} returns the term-groups
\begin{equation}
\mathcal I_0=\{0,1,3\},
\qquad
\mathcal I_1=\{2\}.
\label{eq:construction_running_xyz_groups}
\end{equation}
Generation gives the corresponding sign-flip Pauli strings
\begin{equation}
K_0=XXXX,
\qquad
K_1=XXII.
\label{eq:construction_running_xyz_K}
\end{equation}
It can be directly verified that \(K_0\) anti-commutes with every Pauli term in \(\mathcal I_0\), and \(K_1\) anti-commutes with the unique Pauli term in \(\mathcal I_1\). Therefore the resulting grouped Hamiltonians are 
\begin{equation}
\begin{split}
H_0&=h_0\,XIIZ+h_1\,IZXI+h_3\,YXZY,\\
H_1&=h_2\,XZXZ.
\label{eq:construction_running_xyz_grouped_hamiltonians}
\end{split}
\end{equation}
Section~\ref{sec:grouping} proves that this construction yields a sign-flip Pauli string for every resulting term-group.

\section{Correctness and structural properties of the construction}
\label{sec:grouping}

We now prove the anti-commutation property of the Pauli strings generated by the algorithm and analyze the structure of the resulting term-groups. We first establish the stage-wise invariant that connects Reduction and Generation, and then prove the correctness of the generated sign-flip Pauli strings. We next show two structural invariances of the group label and explain how the resulting term-groups are used in the product-formula implementation. Finally, we summarize the bounds on the number of term-groups, the weights of the generated Pauli strings, and the preprocessing cost.

\subsection{Correctness of the Recursive Construction}
\label{subsec:correctness_grouping_algorithm}

For each positive integer $r$, define the symplectic form on $V^r$ by
\begin{equation}
\begin{split}    
\Omega_r(\bm{u},\bm{v})
&:=
\sum_{t=0}^{r-1} \omega(u_t,v_t),\\
\bm{u}=(u_0,\dots,u_{r-1})\in V^r,
\quad
&\bm{v}=(v_0,\dots,v_{r-1})\in V^r.
\end{split}
\label{eq:grouping_word_symplectic_form}
\end{equation}
When $r=n$, this agrees with the global symplectic inner product in Eq.~\eqref{eq:problem_global_symplectic_form}. When $r=1$, we identify $\Omega_1(u,v)$ with $\omega(u,v)$ for $u,v\in V$.

\begin{lemma}[Stage-wise symplectic invariance]
\label{lem:grouping_stagewise_invariance}
Let $\bm{p}^{(m)}\in V^{2^m}$ be nonzero, and let
\begin{equation}
\bm{p}^{(0)},\bm{p}^{(1)},\dots,\bm{p}^{(m)}
\label{eq:grouping_reduction_sequence}
\end{equation}
be its Reduction sequence, with associated bit string
\begin{equation}
\bm{c}=(c_0,\dots,c_{m-1})\in\mathbb F_2^m.
\label{eq:grouping_control_bits_again}
\end{equation}
Let
\begin{equation}
\bm{k}^{(0)},\bm{k}^{(1)},\dots,\bm{k}^{(m)}
\label{eq:grouping_generation_sequence}
\end{equation}
be the Generation sequence obtained from a group label $(s,\bm{c})$ with $s=\kappa(r(\bm{p}^{(m)}))$. Then, for every $j\in\{0,\dots,m-1\}$,
\begin{equation}
\Omega_{2^{j+1}}\bigl(\bm{p}^{(j+1)},\bm{k}^{(j+1)}\bigr)
=
\Omega_{2^j}\bigl(\bm{p}^{(j)},\bm{k}^{(j)}\bigr).
\label{eq:grouping_stagewise_invariance_eq}
\end{equation}
\end{lemma}

\begin{proof}
Fix $j\in\{0,\dots,m-1\}$ and write
\begin{equation}
\bm{p}^{(j+1)}=[\bm{a}^{(j)};\bm{b}^{(j)}],
\qquad
\bm{a}^{(j)},\bm{b}^{(j)}\in V^{2^j}.
\label{eq:grouping_split_for_proof}
\end{equation}
If $c_j=1$, then $\bm{a}^{(j)}=\bm{b}^{(j)}$, $\bm{p}^{(j)}=\bm{a}^{(j)}$, and $\bm{k}^{(j+1)}=[\bm{k}^{(j)};\bm{0}_{2^j}]$. By definition of $\Omega_{2^{j+1}}$,
\begin{equation}
\begin{split}
\Omega_{2^{j+1}}\bigl(\bm{p}^{(j+1)},\bm{k}^{(j+1)}\bigr)
&=\Omega_{2^{j+1}}\bigl([\bm{a}^{(j)};\bm{a}^{(j)}],[\bm{k}^{(j)};\bm{0}_{2^j}]\bigr)\\
&=\Omega_{2^j}\bigl(\bm{a}^{(j)},\bm{k}^{(j)}\bigr)+\Omega_{2^j}\bigl(\bm{a}^{(j)},\bm{0}_{2^j}\bigr)\\
&=\Omega_{2^j}\bigl(\bm{p}^{(j)},\bm{k}^{(j)}\bigr).
\end{split}
\label{eq:grouping_case_one_invariance}
\end{equation}

If $c_j=0$, then $\bm{p}^{(j)}=\bm{a}^{(j)}+\bm{b}^{(j)}$ and $\bm{k}^{(j+1)}=[\bm{k}^{(j)};\bm{k}^{(j)}]$. Again by definition of $\Omega_{2^{j+1}}$ and bilinearity of $\omega$,
\begin{equation}
\begin{split}
\Omega_{2^{j+1}}\bigl(\bm{p}^{(j+1)},\bm{k}^{(j+1)}\bigr)
&=\Omega_{2^{j+1}}\bigl([\bm{a}^{(j)};\bm{b}^{(j)}],[\bm{k}^{(j)};\bm{k}^{(j)}]\bigr)\\
&=\Omega_{2^j}\bigl(\bm{a}^{(j)},\bm{k}^{(j)}\bigr)+\Omega_{2^j}\bigl(\bm{b}^{(j)},\bm{k}^{(j)}\bigr)\\
&=\Omega_{2^j}\bigl(\bm{a}^{(j)}+\bm{b}^{(j)},\bm{k}^{(j)}\bigr)\\
&=\Omega_{2^j}\bigl(\bm{p}^{(j)},\bm{k}^{(j)}\bigr).
\end{split}
\label{eq:grouping_case_zero_invariance}
\end{equation}
The third equality uses bilinearity over \(\mathbb F_2\).
Thus Eq.~\eqref{eq:grouping_stagewise_invariance_eq} holds in both cases.
\end{proof}

\begin{theorem}[Correctness of Algorithm~\ref{alg:construction_grouping}]
\label{thm:grouping_correctness}
For every term-group $g\in\{0,\dots,G-1\}$ produced by Algorithm~\ref{alg:construction_grouping}, and for every $\ell\in\mathcal I_g$, the generated binary symplectic vector $\bm{k}_g$ satisfies
\begin{equation}
\Omega(\bm{p}_\ell,\bm{k}_g)=1.
\label{eq:grouping_correctness_binary}
\end{equation}
Equivalently,
\begin{equation}
\{K_g,P_\ell\}=0
\qquad
\text{for every }\ell\in\mathcal I_g.
\label{eq:grouping_correctness_pauli}
\end{equation}
Consequently, the term-group Hamiltonian
\begin{equation}
H_g=\sum_{\ell\in\mathcal I_g} h_\ell P_\ell
\label{eq:grouping_correctness_group_hamiltonian}
\end{equation}
satisfies
\begin{equation}
K_g H_g K_g=-H_g.
\label{eq:grouping_correctness_signflip}
\end{equation}
\end{theorem}

\begin{proof}
Fix $g$ and $\ell\in\mathcal I_g$. By Eq.~\eqref{eq:construction_gamma_from_seed_label},
\begin{equation}
\sigma_\ell=\sigma_{[g]}=(s_{[g]},\bm{c}_{[g]}),
\label{eq:grouping_correctness_same_seed_label}
\end{equation}
where $s_{[g]}$ and $\bm{c}_{[g]}$ are a seed symbol and control bit sequence in the label $\sigma_{[g]}$, respectively.
Let
\begin{equation}
\bm{k}_g^{(0)},\bm{k}_g^{(1)},\dots,\bm{k}_g^{(m)}
\label{eq:grouping_correctness_generated_sequence}
\end{equation}
be the Generation sequence obtained from $\sigma_{[g]}$. Applying Lemma~\ref{lem:grouping_stagewise_invariance} repeatedly gives
\begin{equation}
\Omega_N\bigl(\bm{p}_\ell^{(m)},\bm{k}_g^{(m)}\bigr)
=
\Omega_1\bigl(\bm{p}_\ell^{(0)},\bm{k}_g^{(0)}\bigr).
\label{eq:grouping_correctness_iterated_invariance}
\end{equation}
Let $r_\ell:=r(\bm{p}_\ell^{(m)})\in V\setminus\{0\}$. Then $\bm{k}_g^{(0)}=s_{[g]}=\kappa(r_\ell)$, and Eq.~\eqref{eq:construction_kappa_anticommutes} implies
\begin{equation}
\Omega_1\bigl(\bm{p}_\ell^{(0)},\bm{k}_g^{(0)}\bigr)
=
\omega\bigl(r_\ell,\kappa(r_\ell)\bigr)
=
1.
\label{eq:grouping_correctness_base_case}
\end{equation}
Therefore,
\begin{equation}
\Omega_N\bigl(\bm{p}_\ell^{(m)},\bm{k}_g^{(m)}\bigr)=1.
\label{eq:grouping_correctness_extended}
\end{equation}
Since $\bm{p}_\ell^{(m)}$ is the zero-padded version of $\bm{p}_\ell$ and $\bm{k}_g$ is the truncation of $\bm{k}_g^{(m)}$, the padded tail does not contribute to the symplectic inner product. Hence Eq.~\eqref{eq:grouping_correctness_binary} follows from Eq.~\eqref{eq:grouping_correctness_extended}. The equivalence with Eq.~\eqref{eq:grouping_correctness_pauli} follows from Eq.~\eqref{eq:problem_anticommutation_condition}. Finally, linearity gives Eq.~\eqref{eq:grouping_correctness_signflip}.
\end{proof}

This is the property required by Eq.~\eqref{eq:problem_groupwise_evolution_flip} for ancilla-control removal.

\subsection{Structural Properties of Term-groups}
\label{subsec:term_group_structure}

This subsection identifies structural transformations of padded binary-symplectic words that preserve the group label produced by Algorithm~\ref{alg:construction_grouping}. These properties explain how structurally related Pauli patterns can be assigned to the same term-group. When such terms also act on disjoint sets of system qubits, the corresponding elementary Pauli rotations can be placed in a common logical parallel step.

We establish two invariance properties of the recursive label map: interleaved cyclic-shift invariance and block-wise half-swap invariance.

We first consider cyclic shifts of the interleaved binary representation. For a word
\(\bm{p}=(u_0,\dots,u_{2^r-1})\in V^{2^r}\), with \(u_j=(x_j,z_j)\), define \(\mathsf{C}(\bm{p})\in V^{2^r}\) to be the \(V\)-valued word obtained by cyclically shifting the interleaved bit string
\begin{equation}
(x_0,z_0,x_1,z_1,\dots,x_{2^r-1},z_{2^r-1})
\end{equation}
by one bit and then regrouping consecutive bits into pairs. Equivalently, if
\begin{equation}
\mathsf{C}(\bm{p})=(v_0,\dots,v_{2^r-1}),
\label{eq:grouping_structure_T}
\end{equation}
then
\begin{equation}
v_j=((u_{j-1})_Z,(u_j)_X),
\qquad
j=0,\dots,2^r-1,
\label{eq:grouping_structure_T_explicit_clean}
\end{equation}
where \(j-1\) is taken modulo \(2^r\). For \(2^r=2\), this operation maps
\begin{equation}
(x_0,z_0,x_1,z_1)
\longmapsto
(z_1,x_0,z_0,x_1),
\label{eq:grouping_interleaved_shift_example}
\end{equation}
which gives \(v_0=(z_1,x_0)\) and \(v_1=(z_0,x_1)\).

The next lemma shows that one step of Reduction is compatible with this operation.

\begin{lemma}[Compatibility of Reduction with the interleaved cyclic shift]
\label{lem:grouping_shift_reduction_compatibility}
Let \(\bm{p}^{(r)}\in V^{2^r}\) be nonzero, and let \(\bm{p}^{(r-1)}\) be the word obtained from \(\bm{p}^{(r)}\) by one step of Reduction. Set
\begin{equation}
\widetilde{\bm{p}}^{(r)}:=\mathsf{C}(\bm{p}^{(r)}),
\end{equation}
and let \(\widetilde{\bm{p}}^{(r-1)}\) be the word obtained from \(\widetilde{\bm{p}}^{(r)}\) by one step of Reduction. Then the first control bits of \(\bm{p}^{(r)}\) and \(\widetilde{\bm{p}}^{(r)}\) are equal, and
\begin{equation}
\widetilde{\bm{p}}^{(r-1)}=\mathsf{C}(\bm{p}^{(r-1)}).
\label{eq:grouping_structure_shift_reduction_compatibility_clean}
\end{equation}
\end{lemma}

\begin{proof}
Write
\begin{equation}
\bm{p}^{(r)}=(u_0,\dots,u_{2^r-1}),
\qquad
u_j=(x_j,z_j),
\end{equation}
which gives
\begin{equation}
\widetilde{\bm{p}}^{(r)}=(v_0,\dots,v_{2^r-1}),
\qquad
v_j=(z_{j-1},x_j).
\end{equation}
Here \(j-1\) is taken modulo \(2^r\).

The first control bit of \(\bm{p}^{(r)}\) is \(1\) exactly when
\begin{equation}
u_j=u_{j+2^{r-1}}
\qquad
(j=0,\dots,2^{r-1}-1).
\end{equation}
By the explicit formula for \(v_j\), this is equivalent to
\begin{equation}
v_j=v_{j+2^{r-1}}
\qquad
(j=0,\dots,2^{r-1}-1),
\end{equation}
and the first control bit is unchanged.

If this control bit is \(1\), then one step of Reduction keeps the first half of each word. Hence
\begin{equation}
\bm{p}^{(r-1)}=(u_0,\dots,u_{2^{r-1}-1}),
\qquad
\widetilde{\bm{p}}^{(r-1)}=(v_0,\dots,v_{2^{r-1}-1}),
\end{equation}
and Eq.~\eqref{eq:grouping_structure_shift_reduction_compatibility_clean} holds.

If the first control bit is \(0\), then Reduction replaces each pair of corresponding entries by their sum. Thus the \(j\)-th entry of \(\bm{p}^{(r-1)}\) is
\begin{equation}
u_j+u_{j+2^{r-1}}
=
(x_j+x_{j+2^{r-1}},\,z_j+z_{j+2^{r-1}}),
\end{equation}
while the \(j\)-th entry of \(\widetilde{\bm{p}}^{(r-1)}\) is
\begin{equation}
v_j+v_{j+2^{r-1}}
=
(z_{j-1}+z_{j+2^{r-1}-1},\,x_j+x_{j+2^{r-1}}).
\end{equation}
This is the \(j\)-th entry of \(\mathsf{C}(\bm{p}^{(r-1)})\). Therefore Eq.~\eqref{eq:grouping_structure_shift_reduction_compatibility_clean} holds in both cases.
\end{proof}

We can now deduce that cyclic shifts preserve the group label.

\begin{proposition}[Interleaved cyclic-shift invariance for the concrete seed map]
\label{prop:grouping_interleaved_shift_invariance}
Assume that \(\kappa\) is the concrete map in Eq.~\eqref{eq:construction_kappa_choice}. Let \(\bm{p}^{(m)}\in V^N\) be a nonzero padded word. Then \(\bm{p}^{(m)}\) and \(\mathsf{C}(\bm{p}^{(m)})\) produce the same control-bit sequence under Reduction. Moreover, if
\[
r(\bm{p}^{(m)})=(x,z),
\]
then
\begin{equation}
r(\mathsf{C}(\bm{p}^{(m)}))=(z,x).
\label{eq:grouping_structure_shift_root_relation_clean}
\end{equation}
Consequently, \(\bm{p}^{(m)}\) and \(\mathsf{C}(\bm{p}^{(m)})\) produce the same group label.
\end{proposition}

\begin{proof}
Apply Lemma~\ref{lem:grouping_shift_reduction_compatibility} repeatedly through the \(m\) reduction steps. This shows that every control bit is preserved. At the final step, the reduced word has length one, and \(\mathsf{C}\) maps \((x,z)\) to \((z,x)\). Hence Eq.~\eqref{eq:grouping_structure_shift_root_relation_clean} follows.

For the concrete map in Eq.~\eqref{eq:construction_kappa_choice},
\begin{equation}
\kappa((z,x))=\kappa((x,z))
\qquad
((x,z)\in V\setminus\{0\}),
\end{equation}
because \(\kappa((1,0))=\kappa((0,1))=(1,1)\) and \(\kappa((1,1))=(1,0)\). Thus the control-bit sequence and the seed symbol are both unchanged, and the two words produce the same group label.
\end{proof}

As an immediate consequence, any number of repeated cyclic shifts of the interleaved binary representation preserves the group label.

We next consider another operation that preserves the group label, namely exchanging the two halves inside each block of a fixed length.

\begin{definition}[Block-wise half-swap]
\label{def:grouping_blockwise_half_swap}
For each integer \(r\ge 1\) and each \(l\in\{0,\dots,r-1\}\), define
\(\mathsf{B}^{(r)}_l:V^{2^r}\to V^{2^r}\) by
\begin{equation}
\bigl(\mathsf{B}^{(r)}_l(\bm{u})\bigr)_j
=
u_{j\oplus 2^l},
\qquad
0\le j\le 2^r-1,
\label{eq:grouping_structure_blockwise_swap_coordinate}
\end{equation}
where \(\oplus\) denotes bitwise xor of integer indices. This permutation exchanges the two halves of length \(2^l\) in each consecutive block of length \(2^{l+1}\).
\end{definition}

\begin{proposition}[Block-wise half-swap invariance]
\label{prop:grouping_blockwise_swap_invariance}
Let \(\bm{p}^{(m)}\in V^{2^m}\) be nonzero. Then, for every
\(l\in\{0,\dots,m-1\}\), the words \(\bm{p}^{(m)}\) and
\(\mathsf{B}^{(m)}_l(\bm{p}^{(m)})\) produce the same control-bit
sequence and the same root symbol under Reduction. Consequently, they
produce the same group label.
\end{proposition}

\begin{proof}
We prove the statement by induction on \(r\). Let \(\bm{u}\in V^{2^r}\) be nonzero.

First consider \(r=1\). Then \(l=0\). Write
\begin{equation}
\bm{u}=[a;b],
\qquad
a,b\in V .
\end{equation}
The map \(\mathsf{B}^{(1)}_0\) gives
\begin{equation}
\mathsf{B}^{(1)}_0(\bm{u})=[b;a].
\end{equation}
The control bit is unchanged, because \(a=b\) if and only if \(b=a\). If this bit is \(1\), both words reduce to the same root \(a=b\). If this bit is \(0\), the two reduced roots are \(a+b\) and \(b+a\), which are equal. Hence the statement holds for \(r=1\).

Assume that the statement holds for length \(2^{r-1}\), and write
\begin{equation}
\bm{u}=[\bm{a};\bm{b}],
\qquad
\bm{a},\bm{b}\in V^{2^{r-1}} .
\end{equation}

If \(l=r-1\), then
\begin{equation}
\mathsf{B}^{(r)}_{r-1}(\bm{u})=[\bm{b};\bm{a}] .
\end{equation}
The top-level control bit is unchanged. If this bit is \(1\), both words reduce to \(\bm{a}=\bm{b}\). If this bit is \(0\), the first reduced words are \(\bm{a}+\bm{b}\) and \(\bm{b}+\bm{a}\), which are equal. Hence the remaining Reduction steps are identical.

If \(l<r-1\), then \(\mathsf{B}^{(r)}_l\) acts separately on the upper and lower halves:
\begin{equation}
\mathsf{B}^{(r)}_l([\bm{a};\bm{b}])
=
[\mathsf{B}^{(r-1)}_l(\bm{a});
 \mathsf{B}^{(r-1)}_l(\bm{b})] .
\end{equation}
Since \(\mathsf{B}^{(r-1)}_l\) is a coordinate permutation,
\begin{equation}
\bm{a}=\bm{b}
\quad\Longleftrightarrow\quad
\mathsf{B}^{(r-1)}_l(\bm{a})
=
\mathsf{B}^{(r-1)}_l(\bm{b}) .
\end{equation}
Thus the top-level control bit is unchanged.

If this bit is \(1\), then the first reduced words are \(\bm{a}\) and \(\mathsf{B}^{(r-1)}_l(\bm{a})\). Since \(\bm{u}\ne 0\) and \(\bm{a}=\bm{b}\), one has \(\bm{a}\ne 0\). The induction hypothesis applied to \(\bm{a}\) shows that the remaining Reduction outputs are the same.

If this bit is \(0\), then the first reduced words are \(\bm{a}+\bm{b}\) and
\begin{equation}
\mathsf{B}^{(r-1)}_l(\bm{a})
+
\mathsf{B}^{(r-1)}_l(\bm{b})
=
\mathsf{B}^{(r-1)}_l(\bm{a}+\bm{b}).
\end{equation}
Since \(\bm{a}\ne\bm{b}\), one has \(\bm{a}+\bm{b}\ne 0\). The induction hypothesis applied to \(\bm{a}+\bm{b}\) shows that the remaining Reduction outputs are the same.

Therefore, in all cases, \(\bm{u}\) and \(\mathsf{B}^{(r)}_l(\bm{u})\) produce the same control-bit sequence and the same root symbol. Taking \(r=m\) proves the proposition.
\end{proof}

Propositions~\ref{prop:grouping_interleaved_shift_invariance} and~\ref{prop:grouping_blockwise_swap_invariance} identify two structured transformations that preserve the group label. For \(n<N\), the statements are understood on zero-padded words. Pauli patterns generated by these transformations may therefore be merged into one term-group when the corresponding zero-padded words appear in the input Pauli list.

\subsubsection{Example}

As an example, consider \(n=8\), where no zero padding is needed and \(m=3\). Let \(\bm{p}\in V^8\) be the word corresponding to
\begin{equation}
P(\bm{p})=XYIZIIXI .
\end{equation}
Applying the block-wise half-swap \(\mathsf{B}^{(3)}_1\), which swaps the two length-two halves inside each length-four block, gives
\begin{equation}
P(\mathsf{B}^{(3)}_1(\bm{p}))=IZXYXIII .
\end{equation}
Applying \(\mathsf{B}^{(3)}_0\), which swaps the two entries inside each length-two block, gives
\begin{equation}
P(\mathsf{B}^{(3)}_0(\mathsf{B}^{(3)}_1(\bm{p})))=ZIYXIXII .
\end{equation}
Finally, applying \(\mathsf{C}\) four times gives
\begin{equation}
P(\mathsf{C}^4(\mathsf{B}^{(3)}_0(\mathsf{B}^{(3)}_1(\bm{p}))))=IIZIYXIX .
\end{equation}
The initial and final Pauli strings have disjoint supports:
\begin{equation}
\begin{split}
&\operatorname{supp}(XYIZIIXI)=\{0,1,3,6\},\\
&\operatorname{supp}(IIZIYXIX)=\{2,4,5,7\}.
\end{split}
\end{equation}
From Propositions~\ref{prop:grouping_interleaved_shift_invariance}
and~\ref{prop:grouping_blockwise_swap_invariance}, these transformations do not change the group label. Therefore, if \(XYIZIIXI\) and \(IIZIYXIX\) both appear in the input Pauli list, they are assigned to the same term-group although their supports are disjoint. Hence the corresponding Pauli evolutions can be placed in the same phase-rotation layer inside that term-group.

\subsection{Product-formula Implementation of the Grouped Construction}
\label{subsec:grouped_control_free_product_formula}

This subsection describes how the grouped control-free construction is used when the full Hamiltonian evolution is approximated by product formulas. The term-groups and the in-group order are fixed in a first-order ordered kernel. Higher-order formulas, when used, are obtained by recursively composing this fixed kernel and its reversed-order counterpart. No new layer assignment or reordering is applied after expanding the higher-order sequence.
In the following, we explain the details.

Let
\begin{equation}
H=
\sum_{g=0}^{G-1}H_g,
\qquad
H_g=\sum_{\nu=0}^{R_g-1}H_{g,\nu},
\label{eq:grouping_pf_full_decomposition}
\end{equation}
where \(H_g\) is the term-group Hamiltonian in Eq.~\eqref{eq:grouping_correctness_group_hamiltonian}. The refinement \(H_g=\sum_{\nu}H_{g,\nu}\) is a circuit-level decomposition chosen after the term-groups have been fixed, and is not an additional output of Algorithm~\ref{alg:construction_grouping}. Each \(H_{g,\nu}\) is either a single Pauli term or a commuting partial sum implemented as one elementary evolution, and \(R_g\) is the number of such components in the \(g\)th term-group. Let
\begin{equation}
M:=\sum_{g=0}^{G-1} R_g
\label{eq:grouping_pf_total_number}
\end{equation}
be the total number of elementary evolutions. After fixing a total order of these \(M\) components, denote them by \(A_0,A_1,\dots,A_{M-1}\). The first-order ordered kernel is
\begin{equation}
S_1(t):=\prod_{m=0}^{M-1} e^{-\mathrm{i}t A_m} = e^{-\mathrm{i}t A_0} e^{-\mathrm{i}t A_1} \cdots e^{-\mathrm{i}t A_{M-1}}.
\label{eq:grouping_pf_S1}
\end{equation}

The proposed construction is applied to the elementary factors in this kernel. If \(A_m\) belongs to term-group \(g\), then Theorem~\ref{thm:grouping_correctness} gives
\begin{equation}
K_g A_m K_g=-A_m,
\qquad
K_g e^{-\mathrm{i}t A_m}K_g=e^{\mathrm{i}t A_m}.
\label{eq:grouping_pf_factor_sign_flip}
\end{equation}
Thus each elementary evolution can change sign by controlling only \(K_g\), while the evolution factor \(e^{-\mathrm{i}t A_m}\) itself remains uncontrolled.

A higher-order product formula contains repeated occurrences of the elementary evolutions \(e^{-\mathrm{i}t A_m}\), with time steps determined by the recursive composition. Each occurrence inherits the group label and the sign-flip Pauli string assigned to the corresponding factor in \(S_1\). Adjacent identical sign-flip Pauli operations may be cancelled when they become adjacent in the ordered circuit. We do not assign new layers or reorder the fully expanded higher-order sequence, because moving noncommuting factors across recursive blocks would change the Lie--Trotter--Suzuki formula being implemented.

Consequently, the ordering data used by the construction are fixed at the level of \(S_1\). Higher-order time-symmetric formulas are obtained from this fixed kernel and its reversed-order counterpart, as described in Appendix~A. Every occurrence of \(A_m\) in the recursively constructed formula keeps the term-group label assigned in the first-order kernel, and no additional ordering is introduced at the higher-order level.

\subsection{Bounds on Term-groups, Generated Weights, and Preprocessing Cost}
\label{subsec:grouping_complexity}

We now summarize three parameters of the construction: the number of term-groups, the weight of the generated sign-flip Pauli strings, and the classical preprocessing cost. The relevant preprocessing consists of the Reduction stage applied to each input Pauli term and the Generation stage applied to each distinct group label.

First, we state the upper bound on the number of term-groups.

\begin{proposition}[Upper bound on the number of term-groups]
\label{prop:grouping_group_count}
Let $m=\lceil \log_2 n\rceil$ and $N=2^m$. Then the number of term-groups produced by Algorithm~\ref{alg:construction_grouping} satisfies
\begin{equation}
G \le \min\{L,\ |\operatorname{Im}(\kappa)|\,2^m\}.
\label{eq:number_groups_upper_bound_general}
\end{equation}
For the concrete map in Eq.~\eqref{eq:construction_kappa_choice},
\begin{equation}
|\operatorname{Im}(\kappa)|=2,
\label{eq:number_groups_upper_bound_concrete_image}
\end{equation}
and therefore
\begin{equation}
G \le \min\{L,2N\},
\qquad
N=2^{\lceil\log_2 n\rceil}.
\label{eq:number_groups_upper_bound_concrete}
\end{equation}
Since \(n\le N<2n\), this bound is linear in \(n\).
\end{proposition}

\begin{proof}
A group label consists of a symbol in $\operatorname{Im}(\kappa)$ and an $m$-bit string. Hence the total number of formal group labels is at most $|\operatorname{Im}(\kappa)|2^m$. Since at most $L$ of them can actually appear among the $L$ input terms, Eq.~\eqref{eq:number_groups_upper_bound_general} follows. The remaining statements are immediate from Eq.~\eqref{eq:construction_kappa_choice} and the definition of $N$.
\end{proof}

The weight of $K_g$ is relevant to the cost of the controlled Pauli operations that remain in the circuit. Therefore, we next give a proposition about the weight.

\begin{proposition}[Weight of a generated sign-flip Pauli string]
\label{prop:generated_weight}
Let $\sigma=(s,\bm{c})$ be a group label, where $s\in V\setminus\{0\}$ and $\bm{c}=(c_0,\ldots,c_{m-1})\in\mathbb F_2^m$. Define
\begin{equation}
z(\bm{c}):=\left|\{j\in\{0,\ldots,m-1\}\mid c_j=0\}\right|.
\label{eq:zero_count_control_bits}
\end{equation}
Let $\bm{k}^{(m)}$ be the padded word generated from $\sigma$ by Definition~\ref{def:generation_general_pauli}. Then
\begin{equation}
\mathrm{wt}(\bm{k}^{(m)})=2^{z(\bm{c})},
\label{eq:padded_generated_weight}
\end{equation}
where $\mathrm{wt}$ denotes the number of nonzero local $V$-symbols. After truncation to the first $n$ qubits, the resulting sign-flip Pauli string satisfies
\begin{equation}
\mathrm{wt}(K(\sigma))\leq \min\{n,2^{z(\bm{c})}\}.
\label{eq:truncated_generated_weight}
\end{equation}
\end{proposition}

\begin{proof}
The initial word $\bm{k}^{(0)}=s$ has weight one. At the $j$-th Generation step, if $c_j=1$, then $\bm{k}^{(j+1)}=[\bm{k}^{(j)};\bm{0}_{2^j}]$, and the weight is unchanged. If $c_j=0$, then $\bm{k}^{(j+1)}=[\bm{k}^{(j)};\bm{k}^{(j)}]$, and the weight is doubled. This proves Eq.~\eqref{eq:padded_generated_weight}. Truncation can only remove nonzero entries, and the truncated word has length $n$. Therefore Eq.~\eqref{eq:truncated_generated_weight} follows.
\end{proof}

Finally, the classical preprocessing complexity to calculate the term-group labels is summarized as follows.
\begin{proposition}[Classical preprocessing complexity]
\label{prop:grouping_complexity}
Let $N=2^{\lceil \log_2 n\rceil}$. For one input Pauli term, the Reduction stage costs $O(N)$ elementary bit operations. For one distinct group label, the Generation stage also costs $O(N)$. Consequently, the total cost of constructing the term-group labels for all input terms and the sign-flip Pauli strings for all resulting groups is
\begin{equation}
O(LN+GN).
\label{eq:grouping_complexity_total}
\end{equation}
\end{proposition}

\begin{proof}
At each reduction level, the word length is halved. The total work per input term is proportional to
\begin{equation}
N+\frac{N}{2}+\frac{N}{4}+\cdots < 2N.
\label{eq:grouping_complexity_geometric_sum}
\end{equation}
The same bound holds for Generation on one distinct group label. Summing over all $L$ input terms and all $G$ distinct group labels gives Eq.~\eqref{eq:grouping_complexity_total}.
\end{proof}

\section{Numerical Experiments}
\label{sec:numerical}

This section evaluates the proposed construction at the level of the signed controlled time-evolution circuit in Eq.~\eqref{eq:intro_signed_controlled_evolution}. 

\subsection{Compared Constructions, Metrics, and Compilation Settings}
\label{subsec:numerical_scenarios}

In the numerical experiments, we assume that the Hamiltonian-evolution
block is constructed with the first-order ordered product formula. Thus each Pauli term
appears once as an elementary Pauli evolution in the ordered kernel.

Under the setting above, we compare three constructions. The first construction, denoted by $C_{\mathrm{direct}}$, controls every elementary gate in the time-evolution circuit by the ancilla qubit. This is the most direct implementation of the signed controlled time evolution.

The second construction, denoted by $C_{\mathrm{signed\mbox{-}Pauli}}$, implements each elementary signed Pauli evolution directly as a Pauli rotation on the ancilla-extended Pauli string. For a Pauli term $P_\ell$, it uses
\begin{equation}
\ketbra{0}{0}_{\mathcal A}\otimes e^{\mathrm{i}\theta P_\ell}
+
\ketbra{1}{1}_{\mathcal A}\otimes e^{-\mathrm{i}\theta P_\ell}
=
\exp\!\left(\mathrm{i}\theta Z_{\mathcal A}\otimes P_\ell\right).
\label{eq:numerical_signed_pauli_baseline}
\end{equation}
This construction is stronger than controlling every elementary gate, because it treats the signed controlled Pauli evolution as one Pauli rotation. Thus it replaces each signed controlled Pauli evolution by a Pauli rotation generated by the \((n+1)\)-qubit Pauli string \(Z_{\mathcal A}\otimes P_\ell\). The phase rotation in Eq.~\eqref{eq:numerical_signed_pauli_baseline} nevertheless always contains the ancilla qubit. Hence, even if two system Pauli terms have disjoint supports, their ancilla-extended generators overlap at the ancilla. In the support-disjoint layer model used below, the elementary phase rotations of $C_{\mathrm{signed\mbox{-}Pauli}}$ therefore remain serial at the level of phase-rotation layers.

The third construction is the proposed grouped control-free construction. It first applies the recursive grouping algorithm from Sections~\ref{sec:construction} and~\ref{sec:grouping}. For each term-group, the generated anti-commuting Pauli string is applied under ancilla control, while the Pauli evolutions inside the grouped block are not ancilla-controlled. Therefore, if two in-group Pauli terms have disjoint system supports, their elementary phase rotations can be assigned to the same phase-rotation layer. This is the main structural difference between $C_{\mathrm{signed\mbox{-}Pauli}}$ and the proposed construction. It is also the main reason why the proposed construction can reduce the compiled $T$ depth after Clifford+$T$ synthesis of the phase rotations.

After the elementary Pauli-evolution order is fixed, $C_{\mathrm{direct}}$, $C_{\mathrm{signed\mbox{-}Pauli}}$, and the proposed construction place the elementary Pauli-evolution factors according to that same order. For the proposed construction, the phase-rotation layer count is evaluated from support-disjoint layers inside each term-group. Each counted layer contains terms whose system supports are pairwise disjoint.

We report three metrics. The first is compiled $T$ depth. This metric is the main non-Clifford cost metric in the present experiments, motivated by the cost of magic-state-based fault-tolerant implementations of non-Clifford gates~\cite{bravyi2005Universal,fowler2012Surface,gidney2019Efficient}. The second is compiled CX depth, which measures two-qubit entangling-gate depth after transpilation. The third is the phase-rotation layer count. This count measures how much the elementary Pauli phase rotations can be parallelized before Clifford+$T$ synthesis. Under a fixed synthesis tolerance, each arbitrary-angle $R_z$ rotation contributes a comparable amount of non-Clifford depth after Clifford+$T$ synthesis~\cite{kliuchnikov2013Fast,ross2016Optimal,amy2014Polynomial}. Hence the number of serial phase-rotation layers is a useful proxy for compiled $T$ depth. Routing and basis-change gates can affect the final compiled depth. For $C_{\mathrm{signed\mbox{-}Pauli}}$, this count is the number of distinct non-identity Pauli terms, because every phase rotation contains the ancilla. For the proposed construction, this count is the total number of support-disjoint layers inside all term-groups.

Each elementary system Pauli evolution $e^{-\mathrm{i}\theta P}$ is synthesized by basis changes, a CNOT parity-computation chain, one $R_z(2\theta)$ rotation, and uncomputation. The ancilla-extended rotation in Eq.~\eqref{eq:numerical_signed_pauli_baseline} is synthesized by the same rule applied to the Pauli string $Z_{\mathcal A}\otimes P$. All circuits were generated and transpiled with Qiskit 2.3.1~\cite{javadiabhari2024Quantum} using the Clifford+$T$ basis $\{h,s,sdg,x,z,t,tdg,cx\}$, optimization level $1$, and unitary-synthesis tolerance $10^{-8}$. We use two connectivity models: full connectivity and a bidirectional nearest-neighbor two-dimensional grid. The same transpilation settings are used in both connectivity models.

\subsection{Random Hamiltonians}
\label{subsec:numerical_random}

\subsubsection{Simulation Conditions}
\label{subsubsec:numerical_random_conditions}

The first benchmark uses random Pauli-sum Hamiltonians. For each system size $n=4,\dots,16$ and locality bound $k\in\{2,4\}$, we generated $50$ instances. Each instance starts from $8n$ raw Pauli terms. The factor $8$ is an experimental choice used to keep the number of distinct Pauli terms linear in $n$ while producing enough support overlap for grouping effects to be visible. For each raw term, the locality is sampled uniformly from $\{1,\dots,k\}$, the support is sampled uniformly from all possible combinations with the selected locality, each local non-identity Pauli operator is sampled uniformly from $\{X,Y,Z\}$, and the coefficient is sampled uniformly from $[-1,1]$. Duplicate Pauli strings are aggregated, and the coefficient vector is normalized by its $\ell_1$ norm. The same instance set is used for all three constructions.

\subsubsection{Simulation Results}
\label{subsubsec:numerical_random_results}

\begin{figure*}[!t]
\centering
\subfloat[Random 2-local Hamiltonians.]{
\includegraphics[width=0.48\textwidth]{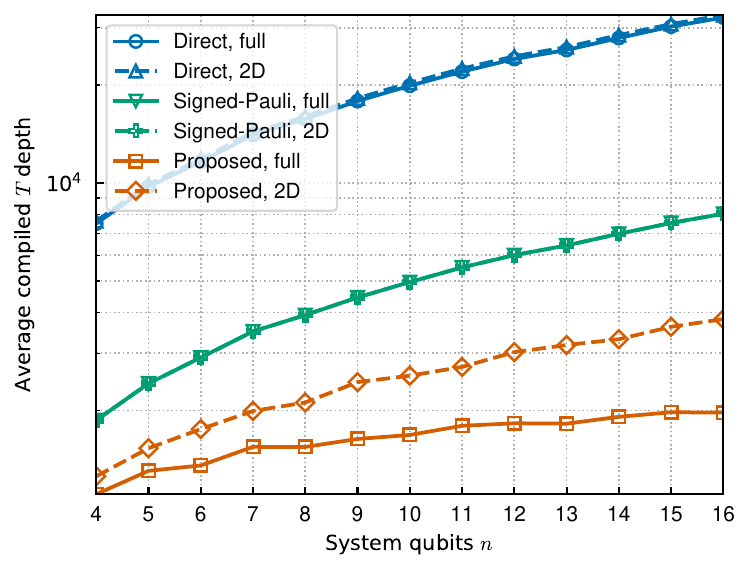}
\label{fig:random_locality2_tdepth}
}
\hfil
\subfloat[Random 4-local Hamiltonians.]{
\includegraphics[width=0.48\textwidth]{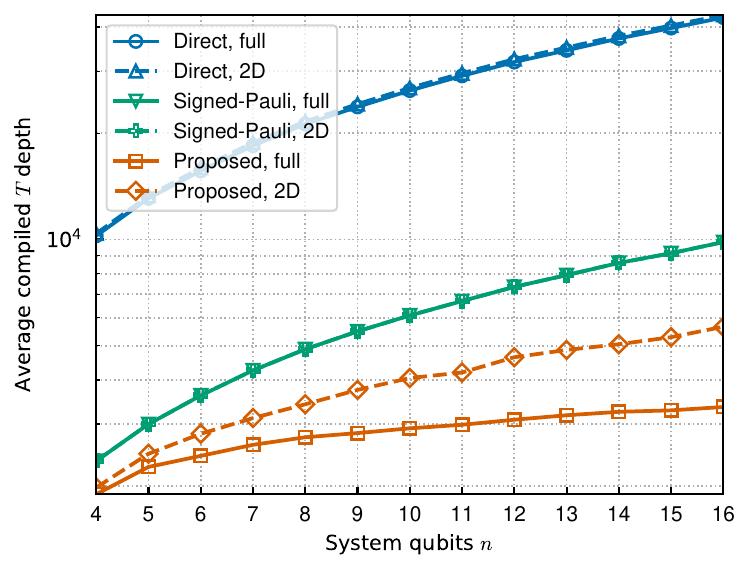}
\label{fig:random_locality4_tdepth}
}
\caption{Average compiled $T$ depth for random Hamiltonians. The direct construction controls every elementary gate in the controlled time-evolution circuit. The signed-Pauli construction implements each signed controlled Pauli evolution as an ancilla-extended Pauli rotation. The proposed construction removes ancilla control from the grouped Pauli-evolution blocks.}
\label{fig:random_tdepth}
\end{figure*}

Figs.~\ref{fig:random_tdepth}(a) and (b) show the average compiled $T$ depth for random 2-local and 4-local Hamiltonians, respectively. In both cases, $C_{\mathrm{signed\mbox{-}Pauli}}$ is already much shallower than $C_{\mathrm{direct}}$ in compiled $T$ depth. The proposed construction further reduces the compiled $T$ depth by exposing phase-rotation parallelism inside the grouped evolution blocks. At $n=16$ under full connectivity, the average compiled $T$ depth for 2-local Hamiltonians decreases from $8042.74$ in $C_{\mathrm{signed\mbox{-}Pauli}}$ to $1970.40$ in the proposed construction, giving a $4.08$-fold reduction. For 4-local Hamiltonians at $n=16$ under full connectivity, it decreases from $9838.06$ in $C_{\mathrm{signed\mbox{-}Pauli}}$ to $3353.34$ in the proposed construction, giving a $2.93$-fold reduction. Under the two-dimensional grid, routing overhead reduces the gap, but the same trend remains.

\begin{figure*}[!t]
\centering
\subfloat[Random 2-local Hamiltonians.]{
\includegraphics[width=0.48\textwidth]{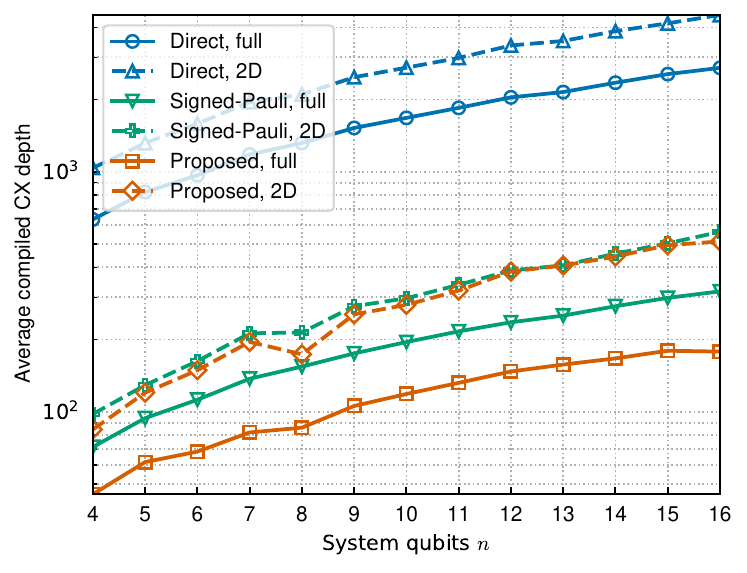}
\label{fig:random_locality2_cxdepth}
}
\hfil
\subfloat[Random 4-local Hamiltonians.]{
\includegraphics[width=0.48\textwidth]{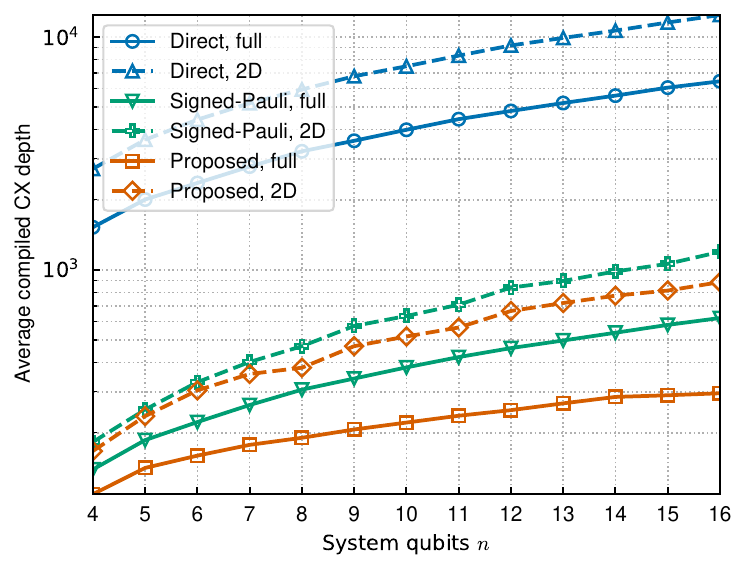}
\label{fig:random_locality4_cxdepth}
}
\caption{Average compiled CX depth for random Hamiltonians under the same settings as in Fig.~\ref{fig:random_tdepth}.}
\label{fig:random_cxdepth}
\end{figure*}

Figs.~\ref{fig:random_cxdepth}(a) and~(b) show the corresponding compiled CX depth for random 2-local and 4-local Hamiltonians, respectively. At $n=16$ under full connectivity, the average compiled CX depth for 2-local Hamiltonians decreases from $317.52$ in $C_{\mathrm{signed\mbox{-}Pauli}}$ to $178.10$ in the proposed construction, giving a $1.78$-fold reduction. For 4-local Hamiltonians at $n=16$ under full connectivity, it decreases from $621.48$ in $C_{\mathrm{signed\mbox{-}Pauli}}$ to $295.48$ in the proposed construction, giving a $2.10$-fold reduction. Under the grid topology, the CX-depth gap is smaller because routing overhead is present for all constructions. Nevertheless, the proposed construction has equal or smaller compiled CX depth than $C_{\mathrm{signed\mbox{-}Pauli}}$ in the tested instances.

\begin{figure*}[!t]
\centering
\subfloat[Random 2-local Hamiltonians.]{
\includegraphics[width=0.48\textwidth]{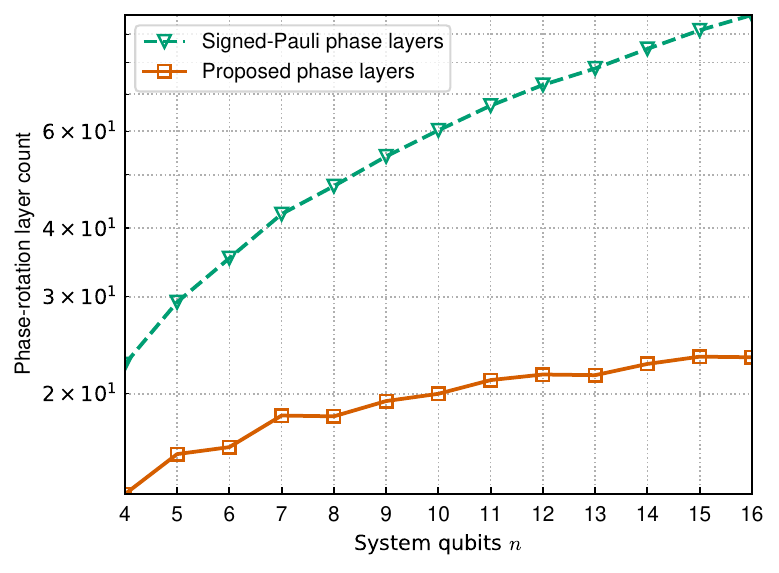}
\label{fig:random_locality2_layers}
}
\hfil
\subfloat[Random 4-local Hamiltonians.]{
\includegraphics[width=0.48\textwidth]{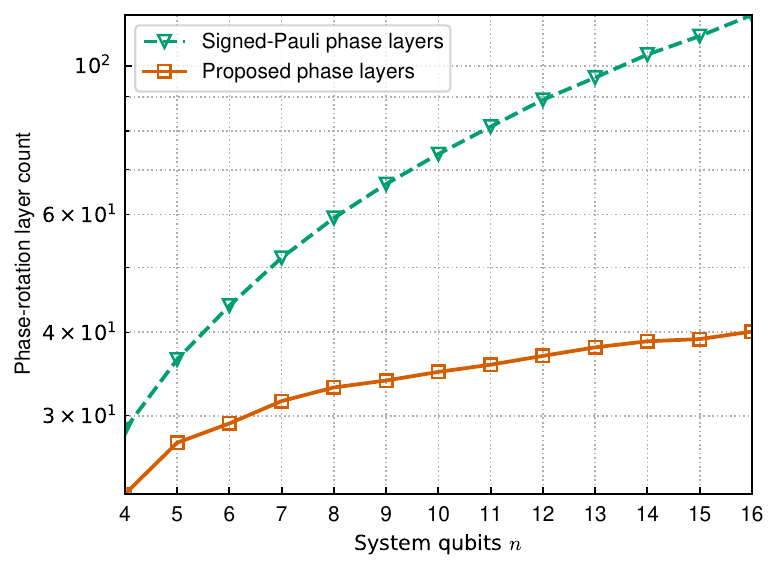}
\label{fig:random_locality4_layers}
}
\caption{Phase-rotation layer count for random Hamiltonians. The signed-Pauli construction has one phase-rotation layer per distinct Pauli term because each phase rotation contains the ancilla. The proposed construction can place support-disjoint in-group Pauli evolutions in the same layer.}
\label{fig:random_layers}
\end{figure*}

Figs.~\ref{fig:random_layers}(a) and~(b) show the phase-rotation layer counts for random 2-local and 4-local Hamiltonians, respectively. These results explain the $T$-depth reduction observed in Figs.~\ref{fig:random_tdepth}(a) and~(b). For 2-local Hamiltonians at $n=16$, the signed-Pauli phase-rotation layer count is $97.52$ on average, while the proposed phase-rotation layer count is $23.30$, giving a $4.19$-fold reduction. For 4-local Hamiltonians at $n=16$, the signed-Pauli phase-rotation layer count is $119.32$ on average, while the proposed phase-rotation layer count is $40.10$, giving a $2.98$-fold reduction. The smaller reduction factor for 4-local Hamiltonians is consistent with the larger support size of each Pauli term. Higher-locality terms overlap on system qubits more often, and fewer pairs of in-group rotations can occupy the same logical phase-rotation layer.

\subsection{Structured Spin Hamiltonians}
\label{subsec:numerical_kagome}

The second benchmark uses a structured spin Hamiltonian on the Kagome lattice. Kagome spin systems are standard examples of geometrically frustrated magnets, and chiral spin liquid phases have been studied in Kagome Mott-insulator models with time-reversal symmetry breaking~\cite{bauer2014Chiral,kalmeyer1987Equivalence,balents2010Spin}. In particular, we use the effective spin model considered by Bauer \textit{et al.}~\cite{bauer2014Chiral} as a structured Pauli-list benchmark.

The Hamiltonian is
\begin{equation}
H_{\rm K}
=
J_{\rm HB}\sum_{\langle i,j\rangle}\mathbf{S}_i\cdot\mathbf{S}_j
+
h_z\sum_i S_i^z
+
J_\chi\sum_{(i,j,k)\in \triangle}\mathbf{S}_i\cdot(\mathbf{S}_j\times\mathbf{S}_k).
\label{eq:kagome_hamiltonian}
\end{equation}
Here, \(\langle i,j\rangle\) denotes a nearest-neighbor edge of the periodic Kagome lattice. The symbol \(\triangle\) denotes the set of elementary triangles formed by three mutually nearest-neighbor spin sites of the lattice. For each \((i,j,k)\in\triangle\), the three sites are ordered clockwise. The geometric convention is specified in Appendix~B.

Following the parametrization of Bauer \textit{et al.}, we set
\begin{equation}
J_{\rm HB}=J\cos\theta,
\qquad
J_\chi=J\sin\theta,
\qquad
J=1.
\label{eq:kagome_parametrization}
\end{equation}
The benchmark parameters are
\begin{equation}
\theta=\frac{\pi}{4},
\qquad
h_z=0.1,
\qquad
t=1.
\label{eq:kagome_parameters}
\end{equation}
Thus $J_{\rm HB}=J_\chi=1/\sqrt{2}$. The choice \(\theta=\pi/4\) gives a balanced Heisenberg and scalar-chirality benchmark. The value \(t=1\) fixes the circuit scale used in all compared constructions. The unnormalized coefficients are generated from Eqs.~\eqref{eq:kagome_hamiltonian}--\eqref{eq:kagome_parameters} and the coefficient vector is then normalized by its $\ell_1$ norm before circuit synthesis.

For the compiled-depth benchmark, we use a periodic $L_x\times L_y=4\times 2$ Kagome lattice with three spins per unit cell. This gives $n=3L_xL_y=24$. This benchmark tests the complete circuit cost on one moderate-size structured instance.

For the scaling benchmark, we use the sequence
\begin{equation}
\begin{split}
(L_x,L_y,n)
=&(2,2,12),(3,2,18),\\
&(4,2,24),(4,3,36),(4,4,48),\\
&(6,4,72),(8,4,96),(8,6,144),(10,6,180).
\label{eq:kagome_scaling_sequence}
\end{split}
\end{equation}
This scaling benchmark has a different purpose from the compiled-depth benchmark. It tests whether the phase-rotation parallelism exposed by the grouping persists as the Kagome lattice grows. For these larger instances, we do not compile the full Clifford+$T$ circuits due to limitations in our computational resources. Instead, we report the phase-rotation layer count, because this can be used as an indicator for estimating the $T$ depth under full connectivity.

\begin{figure}[t]
\centering
\includegraphics[width=\columnwidth]{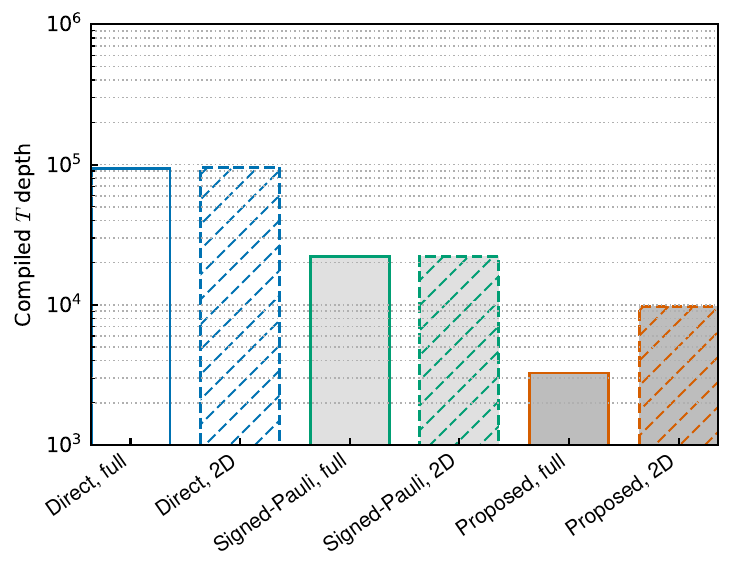}
\caption{Compiled $T$ depth for the structured spin Hamiltonian on the $4\times2$ Kagome lattice, with $n=24$ spins.}
\label{fig:kagome_n24_tdepth}
\end{figure}

\begin{figure}[t]
\centering
\includegraphics[width=\columnwidth]{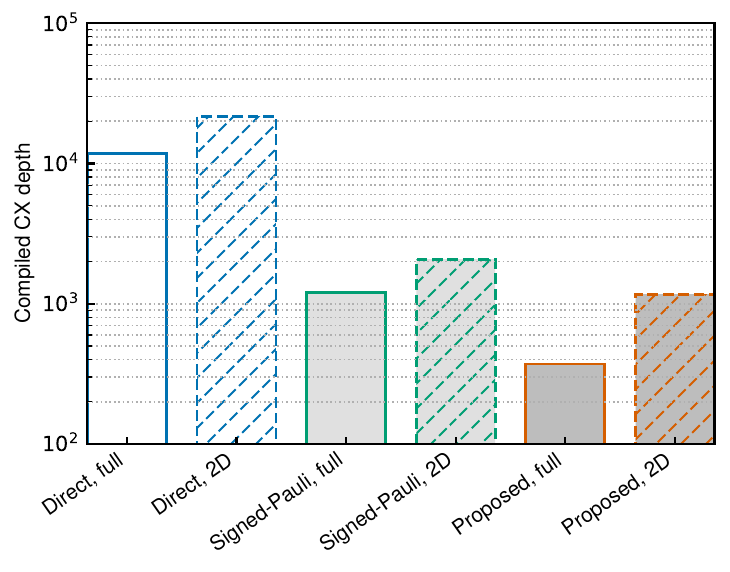}
\caption{Compiled CX depth for the structured spin Hamiltonian under the same settings as in Fig.~\ref{fig:kagome_n24_tdepth}.}
\label{fig:kagome_n24_cxdepth}
\end{figure}

\begin{figure}[t]
\centering
\includegraphics[width=\columnwidth]{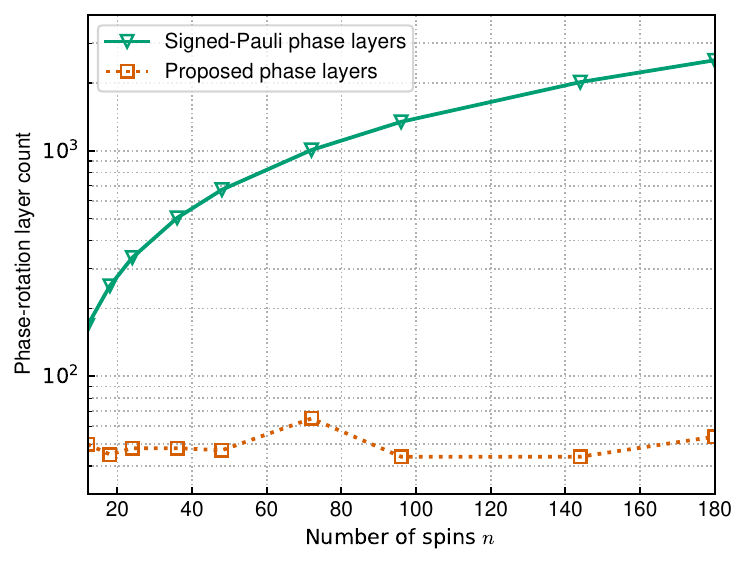}
\caption{Phase-rotation layer count for the structured spin Hamiltonian as a function of the number of spins. The signed-Pauli layer count equals the number of distinct non-identity Pauli terms, while the proposed layer count is obtained from the recursive grouping and support-disjoint in-group layer count.}
\label{fig:kagome_scaling_layers}
\end{figure}

Figs.~\ref{fig:kagome_n24_tdepth} and~\ref{fig:kagome_n24_cxdepth} show the compiled depths for the $n=24$ Kagome instance. Under full connectivity, the compiled $T$ depth decreases from $22104$ in $C_{\mathrm{signed\mbox{-}Pauli}}$ to $3267$ in the proposed construction, giving a $6.77$-fold reduction relative to $C_{\mathrm{signed\mbox{-}Pauli}}$. The direct gatewise construction has compiled $T$ depth $93503$, and the proposed construction gives a $28.62$-fold reduction relative to the direct construction in compiled $T$ depth. Under the two-dimensional grid topology, the compiled $T$ depth decreases from $22104$ in $C_{\mathrm{signed\mbox{-}Pauli}}$ to $9701$ in the proposed construction, giving a $2.28$-fold reduction relative to $C_{\mathrm{signed\mbox{-}Pauli}}$. The corresponding direct gatewise value is $95504$, so the proposed construction gives a $9.84$-fold reduction relative to the direct construction.

The compiled CX-depth data show the same ordering. Under full connectivity, the compiled CX depth decreases from $1198$ in $C_{\mathrm{signed\mbox{-}Pauli}}$ to $372$ in the proposed construction, giving a $3.22$-fold reduction relative to $C_{\mathrm{signed\mbox{-}Pauli}}$. The direct gatewise construction has CX depth $11781$, so the proposed construction gives a $31.67$-fold reduction relative to the direct construction. Under the grid topology, the compiled CX depth decreases from $2074$ in $C_{\mathrm{signed\mbox{-}Pauli}}$ to $1169$ in the proposed construction, giving a $1.77$-fold reduction relative to $C_{\mathrm{signed\mbox{-}Pauli}}$. The direct gatewise construction has CX depth $21759$, so the proposed construction gives an $18.61$-fold reduction relative to the direct construction. Thus the proposed construction reduces both non-Clifford depth and entangling-gate depth in this structured benchmark. The reduction is larger under full connectivity, because the grid topology introduces routing constraints that partially limit the parallelism exposed by the grouped construction.

The phase-rotation layer count gives a structural explanation for the compiled $T$-depth behavior. In the compiled-depth run for the $n=24$ Kagome instance, $C_{\mathrm{signed\mbox{-}Pauli}}$ has $264$ phase-rotation layers because every elementary phase rotation contains the ancilla. The proposed construction has $39$ phase-rotation layers, giving a $6.77$-fold reduction in the layer count. This matches the compiled $T$-depth reduction factor under full connectivity. Therefore, in the fully connected setting, the main source of the compiled $T$-depth reduction is the removal of the ancilla from the grouped evolution blocks, which allows support-disjoint Pauli evolutions inside each term-group to share phase-rotation layers.

Fig.~\ref{fig:kagome_scaling_layers} shows the scaling of the phase-rotation layer count from $n=12$ to $n=180$. In the tested sequence, the number of distinct Pauli terms grows from $132$ at $n=12$ to $1980$ at $n=180$. The proposed phase-rotation layer count stays between $34$ and $55$. At $n=180$, the signed-Pauli construction has $1980$ phase-rotation layers, while the proposed construction has $44$ phase-rotation layers, giving a $45.00$-fold reduction.
Thus, for this sparse structured Hamiltonian, the signed-Pauli layer count grows with the number of Pauli terms, whereas the proposed grouping keeps the phase-rotation layer count nearly size-independent over the tested range. Consequently, the separation between the two constructions becomes larger as the number of qubits increases.

\subsection{Summary of This Section}
\label{subsec:numerical_summary}

The random and structured spin-Hamiltonian benchmarks show the same qualitative behavior. The direct construction has the largest compiled depths. The signed-Pauli construction removes much of the overhead caused by controlling every elementary gate while keeping the ancilla in every elementary phase rotation. The proposed construction removes the ancilla from the grouped Pauli-evolution blocks. This enables support-disjoint Pauli evolutions in the same term-group to share phase-rotation layers, and this phase-rotation layer reduction is reflected in the compiled \(T\) depth.

The influence of the connectivity model should be interpreted separately from this logical layer-count reduction. The phase-rotation layer count is computed before hardware routing. It measures whether elementary Pauli rotations can be placed in the same logical parallel step because their system supports are disjoint. Under full connectivity, this logical parallelism is more directly reflected in the compiled depth. Under the two-dimensional nearest-neighbor grid, a Pauli rotation whose support is not geometrically local requires additional routing operations, such as SWAP insertion and longer parity-computation paths. When several such nonlocal rotations are assigned to the same logical phase-rotation layer, their routed CNOT and SWAP networks can overlap on the physical grid and the logical parallelism can be partially serialized during transpilation. This explains why the proposed construction still reduces compiled \(T\) and CX depths under the grid topology, while the reduction factors are smaller than those under full connectivity. Nevertheless, even under the two-dimensional grid topology, the proposed construction gives the smallest compiled \(T\) depth among the three constructions in the tested benchmarks. Thus the grid topology reduces the relative reduction factor, but not the ordering of the compiled \(T\)-depth results.

\section{Application to Single-Branch Controlled Hamiltonian Evolution}
\label{sec:single_branch}

This section explains how the proposed construction applies to single-branch controlled Hamiltonian evolution. This primitive appears in Hamiltonian simulation by QSP, phase-estimation-type algorithms, GQSP, and Hadamard-test circuits~\cite{low2017Optimal,kitaev1995Quantum,motlagh2024Generalized,berry2024Doubling,cleve1998Quantum}. In these applications, reducing the depth of the controlled Hamiltonian-evolution block reduces the cost of each controlled-query call.

The single-branch controlled operator is
\begin{equation}
V_{H}(t)
=
\ketbra{0}{0}_{\mathcal A}\otimes I
+
\ketbra{1}{1}_{\mathcal A}\otimes e^{-\mathrm{i} t H}.
\label{eq:single_branch_operator}
\end{equation}
A direct implementation of Eq.~\eqref{eq:single_branch_operator} places ancilla control on the Hamiltonian-evolution circuit. For a product-formula realization of \(e^{-\mathrm{i}tH}\), this again introduces controlled elementary Pauli evolutions.

The reduction follows from the signed controlled primitive in Eq.~\eqref{eq:intro_signed_controlled_evolution}. Namely,
\begin{equation}
\begin{split}
\left(I_{\mathcal A}\otimes e^{-\mathrm{i}tH/2}\right)D_H(t/2)
&=
\ketbra{0}{0}_{\mathcal A}\otimes I
+
\ketbra{1}{1}_{\mathcal A}\otimes e^{-\mathrm{i}tH}\\
&=V_H(t).
\end{split}
\label{eq:single_branch_reduction_exact}
\end{equation}
Thus, once \(D_H(t/2)\) is implemented by the proposed grouped control-free construction, the single-branch primitive is obtained by appending an uncontrolled half-step evolution \(I_{\mathcal A}\otimes e^{-\mathrm{i}tH/2}\).

The depth advantage follows from the same mechanism as in the signed controlled case. The Hamiltonian-evolution block is not directly controlled by the ancilla, and the remaining ancilla controls are applied only to the assigned sign-flip Pauli strings. Hence the single-branch primitive can inherit the compiled-depth reduction obtained for the signed controlled primitive, while covering controlled-query blocks used in phase-estimation-type algorithms, generalized quantum signal processing, and Hadamard-test circuits~\cite{kitaev1995Quantum,motlagh2024Generalized,berry2024Doubling,cleve1998Quantum}.

\section{Conclusion}
\label{sec:conclusion}

We studied the construction of signed controlled time-evolution circuits for arbitrary Pauli-sum Hamiltonians. The proposed recursive algorithm partitions the non-identity Pauli terms into term-groups and assigns one Pauli string to each group. The assigned Pauli string anti-commutes with every Pauli term in its group, and therefore flips the sign of the corresponding partial Hamiltonian. This allows the sign of each grouped time-evolution block to be selected by controlling only the assigned Pauli string, while the time-evolution block itself remains uncontrolled.

The construction is formulated in the binary symplectic representation of Pauli strings. The Reduction step assigns a group label to each input Pauli term, and the Generation step maps each distinct label to a sign-flip Pauli string. We proved by a stage-wise symplectic invariant that the generated Pauli string anti-commutes with every Pauli term that has the corresponding group label. For the concrete label map used in this paper, the number of term-groups grows at most linearly with the number of system qubits, and the Reduction and Generation stages require linear preprocessing in the padded word length. We also derived an explicit weight bound for the generated sign-flip Pauli strings. This bound quantifies the number of system qubits touched by the remaining controlled Pauli operations. Fewer zero entries in the control-bit sequence give lower-weight sign-flip strings.

The numerical experiments compared the direct construction, the signed-Pauli construction, and the proposed grouped control-free construction. The proposed construction removes ancilla control from the grouped evolution blocks, allowing support-disjoint Pauli evolutions inside a term-group to share phase-rotation layers. The random-Hamiltonian and structured spin-Hamiltonian benchmarks show that this layer-count reduction is reflected in the compiled \(T\) depth. The experiments also show reductions by our proposed method in compiled CX depth, with a stronger dependence on the connectivity model.

Finally, we showed that the same term-group construction and sign-flip Pauli-string assignment can be used for single-branch controlled Hamiltonian evolution. This follows from a simple algebraic reduction from the signed controlled primitive and extends the construction to one-ancilla controlled-query blocks appearing in various quantum algorithms.

\section*{Appendix}
\subsection{Product-formula recursion and error scaling}
\label{app:product_formula_details}

This appendix explains the Lie--Trotter--Suzuki recursion used in Section~\ref{subsec:grouped_control_free_product_formula}. The product $\prod_{m=0}^{M-1}U_m$ is written in the order $U_0U_1\cdots U_{M-1}$. Let
\begin{equation}
S_1(t):=\prod_{m=0}^{M-1} e^{-\mathrm{i}t A_m}
\end{equation}
be the fixed first-order ordered kernel. Its reversed-order counterpart is
\begin{equation}
S_1^{\mathrm{rev}}(t):=e^{-\mathrm{i}t A_{M-1}}\cdots e^{-\mathrm{i}t A_0}.
\end{equation}
The second-order time-symmetric formula is
\begin{equation}
S_2(t):=S_1(t/2)\,S_1^{\mathrm{rev}}(t/2).
\end{equation}
For each integer $p\ge 2$, the $2p$-th-order Suzuki formula is defined recursively by
\begin{equation}
\begin{split}
S_{2p}(t)
&=
S_{2p-2}(u_pt)^2
S_{2p-2}\bigl((1-4u_p)t\bigr)
S_{2p-2}(u_pt)^2,\\
u_p&:=\frac{1}{4-4^{1/(2p-1)}}.
\end{split}
\end{equation}
For fixed $p$, the repeated formula satisfies the standard asymptotic scaling
\begin{equation}
e^{-\mathrm{i}tH}
=
\bigl[S_{2p}(t/r)\bigr]^r
+
O(t^{2p+1}/r^{2p})
\end{equation}
in the large-$r$ regime \cite{Trotter1959On,suzuki1990Fractal,childs2019Nearly,childs2021Theory}. Since the direct construction, the signed-Pauli construction, and the proposed construction use the same ordered kernel $S_1$ in this paper, this recursion gives the same product-formula order and the same leading product-formula error operator for all three circuit realizations.

\subsection{Geometry of the Structured Spin-Hamiltonian Benchmark}
\label{app:kagome_geometry}

This appendix specifies the Kagome lattice used in Section~\ref{subsec:numerical_kagome}. We follow the orientation convention of the Kagome chiral-spin model of Bauer \textit{et al.}~\cite{bauer2014Chiral}.

We consider an \(L_x\times L_y\) Kagome cluster with periodic boundary conditions. Each unit cell contains three spin sites. A site is labeled by
\begin{equation}
(x,y,s),\qquad
0\le x<L_x,\quad
0\le y<L_y,\quad
s\in\{0,1,2\}.
\label{eq:kagome_site_label}
\end{equation}
The number of spins is
\begin{equation}
n=3L_xL_y.
\label{eq:kagome_n_spins}
\end{equation}
The qubit index assigned to \((x,y,s)\) is
\begin{equation}
q(x,y,s)=3(yL_x+x)+s.
\label{eq:kagome_qubit_index}
\end{equation}

The geometric embedding is fixed by the primitive vectors
\begin{equation}
a_1=(1,0),\qquad
a_2=\left(\frac{1}{2},\frac{\sqrt{3}}{2}\right),
\label{eq:kagome_primitive_vectors}
\end{equation}
and the sublattice offsets
\begin{equation}
\delta_0=(0,0),\qquad
\delta_1=\left(\frac{1}{2},0\right),\qquad
\delta_2=\left(\frac{1}{4},\frac{\sqrt{3}}{4}\right).
\label{eq:kagome_offsets}
\end{equation}
Before imposing periodicity, the position of site \((x,y,s)\) is
\begin{equation}
r_{x,y,s}=x a_1+y a_2+\delta_s.
\label{eq:kagome_position}
\end{equation}
Periodic boundary conditions are imposed with period vectors \(L_xa_1\) and \(L_ya_2\). Nearest-neighbor edges are the pairs of sites whose minimum-image distance is \(1/2\). Elementary triangles are the three-site nearest-neighbor cliques of this periodic graph, and their orientation is fixed by clockwise ordering in the local geometric embedding.

For every cluster used in Section~\ref{subsec:numerical_kagome}, this construction gives the number of edges and triangles as 
\begin{equation}
N_{\rm edge}=6L_xL_y,\qquad
N_{\triangle}=2L_xL_y,
\label{eq:kagome_counts}
\end{equation}
respectively.
The Hamiltonian used in the benchmark is
\begin{equation}
H_{\rm K}
=
J_{\rm HB}\sum_{\langle i,j\rangle}\mathbf{S}_i\cdot\mathbf{S}_j
+
h_z\sum_i S_i^z
+
J_\chi\sum_{(i,j,k)\in \triangle}
\mathbf{S}_i\cdot(\mathbf{S}_j\times\mathbf{S}_k),
\label{eq:appendix_kagome_hamiltonian}
\end{equation}
where \(\mathbf{S}_i=(X_i,Y_i,Z_i)/2\). The three-spin summation uses the clockwise orientation specified above. The benchmark parameters are
\begin{equation}
\begin{split}
&J_{\rm HB}=J\cos\theta,\qquad
J_\chi=J\sin\theta,\qquad
J=1,\\
&\theta=\frac{\pi}{4},\qquad
h_z=0.1.
\label{eq:appendix_kagome_parameters}
\end{split}
\end{equation}
Thus \(J_{\rm HB}=J_\chi=1/\sqrt{2}\). For an oriented triangle \((i,j,k)\), the scalar spin chirality is expanded as
\begin{equation}
\begin{split}
\mathbf S_i\cdot(\mathbf S_j\times\mathbf S_k)
=
\frac{1}{8}\bigl(&X_iY_jZ_k+Y_iZ_jX_k+Z_iX_jY_k\\
&-X_iZ_jY_k-Z_iY_jX_k-Y_iX_jZ_k\bigr).
\end{split}
\label{eq:appendix_scalar_chirality_pauli_expansion}
\end{equation}
Each oriented triangle therefore contributes six three-qubit Pauli strings. Since \(N_{\rm edge}=6L_xL_y=2n\), the Heisenberg term contributes \(3N_{\rm edge}=6n\) two-qubit Pauli strings. Since \(N_\triangle=2L_xL_y=2n/3\), the scalar-chirality term contributes \(6N_\triangle=4n\) three-qubit Pauli strings. Together with the \(n\) Zeeman strings, the total number of non-identity Pauli terms is
\begin{equation}
L=11n
\label{eq:kagome_pauli_term_count}
\end{equation}
for the parameter choice in Eq.~\eqref{eq:appendix_kagome_parameters}. For the \(4\times2\) cluster used in the compiled-depth benchmark, \(n=24\) and \(L=264\).


\bibliographystyle{IEEEtran}
\bibliography{main}

\makeatletter
\@ifundefined{c@biography}{\newcounter{biography}}{}
\makeatother

\newpage
\begin{IEEEbiographynophoto}
{Shintaro~Fujiwara} (Graduate Student Member, IEEE) received the B.E. and M.E. degrees in Engineering Science from Yokohama National University, Kanagawa, Japan, in 2024 and 2025, respectively. He is currently pursuing the Ph.D. degree with the Graduate School of Engineering Science, Yokohama National University, Kanagawa, Japan, and is a Research Fellow (DC1) of the Japan Society for the Promotion of Science (JSPS). His research interests include quantum algorithms, coding theory, and wireless communications.
\end{IEEEbiographynophoto}

\begin{IEEEbiographynophoto}
{Naoki~Yamamoto} is a Professor at the Department of Applied Physics
and Physico-Informatics, Keio University, and the chair of Keio Quantum
Computing Center. He received the B.S. degree in engineering and the M.S.
and Ph.D. degrees in information physics and computing from the University
of Tokyo in 1999, 2001, and 2004, respectively. He was a postdoctoral fellow
at the California Institute of Technology from 2004 to 2007, and at the
Australian National University from 2007 to 2008. His research interests
include quantum computation and control.
\end{IEEEbiographynophoto}

\begin{IEEEbiographynophoto}
{Naoki~Ishikawa} (Senior Member, IEEE) is an Associate Professor with the Faculty of Engineering, Yokohama National University, Kanagawa, Japan. He received the B.E., M.E., and Ph.D. degrees in Electronic and Information Engineering from the Tokyo University of Agriculture and Technology, Tokyo, Japan, in 2014, 2015, and 2017, respectively. In 2015, he was an academic visitor with the School of Electronics and Computer Science, University of Southampton, UK. In 2025, he was a research scholar with the Jacobs School of Engineering, the University of California, San Diego, US. From 2016 to 2017, he was a research fellow of the Japan Society for the Promotion of Science. His research interests include quantum algorithms and wireless communications.
\end{IEEEbiographynophoto}

\end{document}